\newif\ifarxiv
\arxivtrue

\newcommand{\thetitle}{%
  Scheduling Tasks towards Energy Autarky:\\
  Benefits and Computational Costs of Flexibility
}
\newcommand{\thertitle}{%
  Scheduling Tasks towards Energy Autarky%
}
\def\PACS{FL~1247/1-1, 522475669}
\def\thxPACS{Funded by the Deutsche Forschungsgemeinschaft (DFG, German
Research Foundation), project ``Parameterized Algorithmics in Computational Sustainability (PACS)''---\PACS.}
\def\HUaffil{Humboldt-Universität zu Berlin, 
Department of Computer Science, Algorithm Engineering Group, Germany}

\def\myAckn{%
  We gratefully acknowledge 
  Andreas Reinhardt 
  and Mazen Bouchur 
  (TU Clausthal, Energy Informatics group)  
  for their continuous guidance throughout the project
  regarding details and specifics about energy systems.
}

\ifarxiv%
  \documentclass[a4paper, UKenglish, cleveref, autoref, thm-restate]{lipics-v2021}
\else
  \documentclass[a4paper, UKenglish, cleveref, autoref, thm-restate]{socg-lipics-v2021}
\fi
\ifarxiv%
  \hideLIPIcs  %
\fi%

\title{\thetitle} %

\titlerunning{\thertitle} %

\author{Robert Bredereck}{Institut für Informatik, TU Clausthal, Germany}{robert.bredereck@tu-clausthal.de}{https://orcid.org/0000-0002-6303-6276}{}%

\author{Till Fluschnik}{\HUaffil}{till.fluschnik@hu-berlin.de}{https://orcid.org/0000-0002-1825-0097}{\thxPACS}%

\author{Klaus Heeger}{Department of Industrial Engineering and Management, Ben-Gurion University of the Negev, Beer-Sheva, Israel}{heeger@post.bgu.ac.il}{https://orcid.org/0000-0001-8779-0890}{}%

\authorrunning{R.\ Bredereck, T.\ Fluschnik, K.\ Heeger} %

\Copyright{Robert Bredereck and Till Fluschnik and Klaus Heeger} %



\category{} %

\relatedversion{} %

\acknowledgements{\myAckn}%

\ifarxiv{}
  \nolinenumbers %
\fi{}

\EventEditors{John Q. Open and Joan R. Access}
\EventNoEds{2}
\EventLongTitle{42nd Conference on Very Important Topics (CVIT 2016)}
\EventShortTitle{CVIT 2016}
\EventAcronym{CVIT}
\EventYear{2016}
\EventDate{December 24--27, 2016}
\EventLocation{Little Whinging, United Kingdom}
\EventLogo{}
\SeriesVolume{42}
\ArticleNo{23}

\usepackage{adjustbox}
\usepackage{colortbl}
\usepackage{multirow}
\usepackage[textsize=footnotesize,backgroundcolor=green!40!white,linecolor=black!60,obeyFinal]{todonotes}

\usepackage{etoolbox}

\usepackage[capitalise]{cleveref}
\crefname{algocf}{Algorithm}{Algorithms}
\Crefname{algocf}{Algorithm}{Algorithms}
\crefname{axiom}{}{}
\creflabelformat{axiom}{(#2I#1#3)}

\usepackage{url}
\usepackage{dsfont, microtype}
\usepackage{booktabs, tabularx}

\usepackage{tikz}
\usetikzlibrary{arrows, arrows.meta, calc, decorations.pathmorphing, decorations.pathreplacing, calligraphy, backgrounds, math, shapes, shapes.geometric, spy, positioning, quotes}
\usepackage{pgfplotstable}
\usepackage{pgfplots}

\usepackage{xargs}
\newcommandx{\mydefenv}[4][3=A]{%
  \ifstrequal{#3}{A}{\crefname{#1}{#2}{#2s}}{\crefname{#1}{#2}{#3}}%
  \Crefname{#1}{#4{.}}{#4s{.}}
}

\newcommandx{\mydefenvnew}[4][3=A]{%
  \newtheorem{#1}{#2}
  \ifstrequal{#3}{A}{\crefname{#1}{#2}{#2s}}{\crefname{#1}{#2}{#3}}%
  \Crefname{#1}{#4{.}}{#4s{.}}
}

\theoremstyle{plain}
\mydefenv{theorem}{Theorem}{Thm}
\mydefenv{lemma}{Lemma}{Lem}
\mydefenv{proposition}{Proposition}{Prop}
\mydefenv{conjecture}{Conjecture}{Conj}
\mydefenv{corollary}{Corollary}[Corollaries]{Cor}
\mydefenv{observation}{Observation}{Obs}
\mydefenv{fact}{Fact}{Fact}
\theoremstyle{definition}
\mydefenv{definition}{Definition}{Def}
\mydefenv{construction}{Construction}{Constr}
\mydefenvnew{rrule}{Reduction Rule}{DRR}
\mydefenvnew{myalgo}{Algorithm}{Algo}
\mydefenvnew{problem}{Problem}{Prob}
\theoremstyle{remark}
\mydefenv{example}{Example}{Ex}
\mydefenv{remark}{Remark}{Rem}
\mydefenv{idea}{Idea}{Idea}

\newcommand{\rqed}{\hfill$\triangleleft$}
\newcommand{\cif}{\text{if~}}
\newcommand{\cotw}{\text{otherwise}}

\newcommand{\prob}[1]{\textnormal{\textsc{#1}}}
\newcommand{\wpb}{when parameterized by}
\newcommand{\RD}{$(\Rightarrow)\quad$}
\newcommand{\LD}{$(\Leftarrow)\quad$}
\newcommandx{\set}[2][1=1]{\ensuremath{\{#1,\ldots,#2\}}}
\newcommandx{\decprob}[6][3=Input,5=Question]{
\begin{problem}[{#1}]\label{prob:#2}
  Given #4,
  the question is whether #6.
  \end{problem}
}
\newcommandx{\dectask}[6][3=Input,5=Question]{
\begin{problem}[{#1}]\label{prob:#2}
  \textbf{Given} #4,
  the \textbf{task} is to #6.
  \end{problem}
}

\usepackage{mathtools}

\newcommand{\mathset}[1]{\mathbb{#1}}
\newcommand{\N}{\mathset{N}}
\newcommand{\Nzero}{\mathset{N}_0}
\newcommand{\Q}{\mathset{Q}}
\newcommand{\Qnn}{\mathset{Q}_{\geq 0}}
\newcommand{\Qp}{\mathset{Q}_{+}}
\newcommand{\calJ}{\mathcal{J}}
\newcommand{\calN}{\mathcal{N}}
\newcommand{\calR}{\mathcal{R}}

\newcommand{\cocl}[1]{\textrm{#1}}
\newcommand{\XP}{\cocl{XP}}
\newcommand{\W}[1]{\cocl{W[#1]}}
\newcommand{\FPT}{\cocl{FPT}}
\newcommand{\NP}{\cocl{NP}}
\newcommand{\classP}{\cocl{P}}
\newcommand{\coNP}{\cocl{coNP}}
\newcommand{\cpoly}{\cocl{poly}}
\newcommand{\NPincoNPslashpoly}{\ensuremath{\NP\subseteq\coNP/\cpoly}}
\newcommand{\unlessPK}{unless \NPincoNPslashpoly}
\newcommand{\UnlessPK}{Unless \NPincoNPslashpoly}
\newcommand{\croco}{cross-composition}
\newcommand{\ANDcroco}{AND-\croco}

\newcommand{\yes}{\emph{yes}}
\newcommand{\no}{\emph{no}}
\newcommand{\ceq}{\coloneqq}
\newcommand{\maxz}{\max\nolimits^{0}}
\newcommand{\minz}{\min\nolimits^{0}}
\newcommand{\lin}{\lambda_{\rm{in}}}
\newcommand{\lout}{\lambda_{\rm{out}}}
\newcommand{\Jset}{\mathcal{J}}
\newcommand{\Net}{\calN}
\newcommand{\ex}{X}
\newcommand{\Bat}{B}
\newcommand{\bz}{b_0}
\newcommand{\bc}{b_{\mathrm c}}
\newcommand{\bl}{b_{\mathrm \ell}}

\newcommandx{\tref}[2][1=]{\scriptsize{(\Cref{#2}#1)}}
\DeclareMathOperator{\poly}{poly}
\DeclareMathOperator{\im}{im}
\renewcommand{\Im}{\operatorname{Im}}
\newcommand{\indic}{\ensuremath{\mathbf{1}}}

\newcommand{\etal}{et~al.}

\newcommand{\qeTsc}{\prob{Autarky by Scheduling}}
\newcommand{\qeAcr}{\prob{AbS}}
\newcommand{\ubpTsc}{\prob{Unary Bin Packing}}
\newcommand{\ubpAcr}{\prob{UBP}}
\newcommand{\partTsc}{\prob{Partition}}

\newcommand{\nbin}{k}
\newcommand{\sbin}{b}

\newcommand{\pvfnt}[1]{\texttt{#1}}
\newcommand{\jis}{\widehat{\pi}}

\newcommand{\xcase}[2]{\emph{Case~#1}: #2.}

\usepackage{siunitx}
\newcommand{\vwatt}[1]{\SI{#1}{\watt}}
\newcommand{\vsqm}[1]{\SI{#1}{\meter\squared}}
\newcommand{\vmin}[1]{\SI{#1}{\minute}}
\newcommand{\vwattpersqm}[1]{\SI{#1}{\watt\per\meter\squared}}
\newcommand{\vwattmin}[1]{\SI{#1}{\watt\minute}}
\newcommand{\vkilowattmin}[1]{\SI{#1}{\kilo\watt\minute}}
\newcommand{\vkilowatthour}[1]{\SI{#1}{\kilo\watt\hour}}
\newcommand{\vsec}[1]{\SI{#1}{\second}}
\newcommand{\vhour}[1]{\SI{#1}{\hour}}

\newcommand{\cco}{\operatorname{cco}}
\newcommand{\mccco}{\text{m-$c$-$\cco$}}
\newcommand{\mx}[1]{#1}
\newcommand{\nrl}{\rho}
\newcommand{\ndl}{\delta}
\newcommand{\flex}{\phi}

\newcommand{\battsymb}{%
  \ensuremath{%
    \mathord{%
      \tikz[scale=0.025, baseline=0.02ex]{
        \clip (-3,0) rectangle (3,6.5);
        \draw[fill=black] (-1,6.5) rectangle (1,5);
        \draw[draw=black,thin] (-2,0) rectangle (2,5);
        \draw[very thin] (-3,0) -- (3,5.3);
        \draw[very thin] (3,0) -- (-3,5.3);
      }%
    }%
  }%
}

\newcommand{\appsymb}{{\large$\star$}}
\newcommand{\appref}[1]{{\hyperref[proof:#1]{\appsymb}}}

\newcommand{\appendixsection}[1]{%
  \gappto{\appendixProofText}{\section{Additional Material for Section~\ref{#1}}\label{app:#1}}
}

\newcommand{\toappendix}[1]{%
  \gappto{\appendixProofText}
  {{
    #1
  }}
}
\newcommand{\appendixproof}[2]{%
  \gappto{\appendixProofText}{%
    \subsection{%
      \texorpdfstring{Proof of \cref{#1}}{Proof of \ref*{#1}}%
    }\label{proof:#1}
    #2
  }%
}
\newcommand{\theabstract}{%
  We study the \emph{autarky} problem:
  given an energy forecast,
  a battery,
  and a set of energy-consuming jobs with time windows,
  decide whether all jobs can be scheduled
  without requiring
  external energy.
  We analyze the problem through the lens of job flexibility, 
  defined as the number of time steps at which a job may be scheduled.
  We show that the problem is NP-hard already for flexibility two,
  even in restricted settings.
  On the positive side,
  we identify settings in which the problem is polynomial-time solvable,
  even for large flexibilities.
  Moreover,
  we obtain fixed-parameter tractability for combined parameters involving flexibility,
  such as the number of jobs.
  In contrast,
  we establish W-hardness when parameterized by maximum flexibility alone, even in a restricted setting.
  To complement our theoretical results, 
  we formulate an integer linear program (ILP) that computes the minimum required external energy and 
  evaluate it experimentally on instances derived from real-world energy-consumption and radiation data. 
  The experiments indicate that increased job flexibility substantially reduces the need for external energy at moderate computational cost.
}

\newcommand{\plotpath}{pics/}

\begin{document}

\maketitle

  \begin{abstract}
  \theabstract{}
  \end{abstract}
\newpage
\section{Introduction}

Renewable,
weather-dependent resources such as solar and wind
become increasingly important for energy production.
This, in turn, 
increases the importance of reliable forecasts.
Given a reliable forecast of resource availability,
energy-consuming jobs
(abstracting tasks or devices, e.g., a kettle
or a PC)
can be scheduled accordingly to directly consume available power.
This applies both at a local level (e.g., solar panels on private properties)
and at a global level (e.g., offshore wind farms for industrial use).
In this work,
we study the following problem:
given an energy forecast, 
a battery, 
and a set of energy-consuming jobs,
how can the jobs be scheduled
so as to minimize the required external energy---and,
in particular,
to decide whether no external energy is required at all.
The difficulty stems from the interaction between cumulative energy constraints over time
and execution-window constraints, 
which together create long-range dependencies between scheduling decisions.

Our problem is closely linked to the evaluation
of a household's degree of \emph{autarky}.
Several
approaches to this already exist.
What is novel in our approach is that
we optimally solve the underlying scheduling problem
arising from the fact that jobs often allow some flexibility
in their execution.
In particular,
we show that while flexibility makes the problem computationally hard,
it can yield significant energy savings in practice.

The application domain ranges from
households
to quarters 
(i.e., neighborhoods)
and industrial settings,
which changes the view on the problem parameters.
For a household,
we expect relatively few jobs;
for quarters, somewhat more; 
and
for industry,
potentially many.
Conversely,
households and quarters may involve diverse jobs,
while industrial jobs may be more homogeneous.
Moreover,
households and quarters may be less flexible,
while industry may allow more flexibility in execution times, 
focusing primarily on whether a job is completed at all.
Hence,
to fundamentally understand the complexity of the problem,
we perform a classic computational and parameterized complexity analysis.

\subparagraph*{Our Contributions.}
We present an elaborate mathematical model and introduce the energy autarky problem 
\qeTsc{} (\qeAcr).
We conduct an extensive computational and parameterized complexity analysis
(see \cref{fig:results,fig:hasse};
\cref{fig:organigram} organizes our results;
\cref{sec:basics,sec:mpdp} survey our notations)
and show that \qeAcr{} becomes \NP-hard already when every job has flexibility exactly two,
even under further restrictions,
such as each job having unit length.
Moreover,
we show that \qeAcr{} is fixed-parameter tractable regarding the combined parameter flexibility and the order of the largest connected component of the job graph,
which intuitively captures the dependency between jobs.
In contrast, 
we show that \qeAcr{} is \W{1}-hard \wpb{} the number of time steps,
and hence by flexibility,
even if all jobs have length one and the same release date and deadline.
Given the intractability of \qeAcr{},
we formulate an exact algorithm in form of an integer linear program (ILP).
With the ILP at hand,
we run experiments on real-world data 
(household energy consumption and radiation profiles), 
artificially combined into a total of 48048 instances,
all of which are solved using the ILP.
Our results show how job flexibility relates to the minimum required external energy:
In a nutshell,
larger flexibility can save a significant amount of energy at a,
on average,
moderate increase of runtime.
The experimental evaluation is intended to complement the theoretical analysis by illustrating the practical impact of flexibility, 
rather than to serve as a benchmarking study.

\begin{figure*}[t]
 \centering
 \begin{tikzpicture}[font=\footnotesize]
  \def\xr{1}
  \def\yr{0.75}
  \def\ymx{3.75}
  \newcommandx{\resEntry}[6][6=0]{%
    \node at ($(#1)+(#2*\xr,#3*\yr)$)[anchor=north west,xshift=0.4*\xr pt,yshift=-2.1*\yr pt,align=left,text width=#4*\xr cm]{$\bullet$};
    \node at ($(#1)+(#2*\xr,#3*\yr)$)[anchor=north west,xshift=8*\xr pt,yshift=#6*\yr pt,align=left,text width=#4*\xr cm]{#5};
  }

  \def\xOne{0.6}
  \def\xTwo{5.5}
  \def\xThr{13.625}
  
  \draw[->,>=latex] (0,0) to (\xThr*\xr,0);
  
  \foreach \x/\y in {0/0,\xOne/1,\xTwo/2,\xThr/3}{
    \node (n\y) at (\x*\xr,0.125)[inner sep=-1pt]{};
    \ifnum\y<3 \draw[-] (\x*\xr,0.125) -- (\x*\xr,-0.125); \fi
  }
  
  \node at ($(n0)!0.5!(n1)$)[anchor=north, yshift=-6*\yr pt]{$\mx{\flex}=1$};
  \node at ($(n1)!0.5!(n2)$)[anchor=north, yshift=-5*\yr pt]{$\mx{\flex}=2$};
  \node at ($(n2)!0.5!(n3)$)[anchor=north, yshift=-5*\yr pt]{unbounded $\mx{\flex}$};

  \draw[fill=green!15!white, draw=none] (n0) to ($(n0)+(0,\ymx*\yr)$) to ($(n3)+(0,\ymx*\yr)$) to (n3) to (n0);

  \draw[fill=red!17!white, draw=none] (n1) rectangle ($(n2)+(0,1.5*\yr)$);

  \draw[fill=red!25!white, draw=none] (n2) rectangle ($(n3)+(0,2.25*\yr)$);
  \node at ($(n0)!0.5!(n1)+(0,1.0*\yr)$)[]{};
  
  \def\wOne{4.4}

  \resEntry{n1}{0}{0.75}{\wOne}{even if $\nrl=n$, $\lambda=1$
  \tref{prop:horiz:joborders}}
  \resEntry{n1}{0}{1.5}{\wOne}{even if $\tau=2$, $\mx{\ell}=1$ $^{\battsymb,\dagger}$
  \tref{prop:NPhard:Part1}}[3]
  \resEntry{n1}{0}{2.25}{\wOne}{when $\mx{\ell}=1,\nrl=n$  \tref{prop:oneortwouniquerels}}
  \resEntry{n1}{0}{3}{\wOne}{for constant $n$ {\tref[, FPT]{obs:FPTnPlusFlex}}}
  
  \def\wTwo{7.75}
  
  \resEntry{n2}{0}{0.75}{\wTwo}{even if~$e_i=1 \,\forall i\in N$ $^{\battsymb,1}$
  \tref{prop:whard}}[3]
  \resEntry{n2}{0}{1.5}{\wTwo}{even if~$\mx{\ell}=1$, $\nrl=n$ $^{\battsymb}$
  \tref{prop:NPhard:BP:rdiff}}[3]
  \resEntry{n2}{0}{2.25}{\wTwo}{even if $\mx{\ell}=\nrl=\ndl=1$ $^{\battsymb,\dagger,2}$ 
  \tref{prop:NPhard:BP:rsame}}[3]
  \resEntry{n2}{0}{3}{\wTwo}{when~$\mx{\ell}=1$, $e_i=e_j$, $\flex_i=\flex_j \,\forall i,j\in N$, $\bc=\bl$  \tref{prop:greedy}}
  \resEntry{n2}{0}{3.75}{\wTwo}{for constant $n$ when no battery exists \tref[, FPT]{prop:nobattery:FPTn}}

 \end{tikzpicture}
 \caption{Overview of our results for~\qeAcr{} with $\lambda=\min\{\lin,\lout\}$ 
 (green: P-time; light/darker-red: weak/strong \NP-hardness).  
 Some results hold even if 
 $^{\battsymb}$: there is no battery;
 $^\dagger$: the forecast is the same for all time steps.
 For some restrictions we also have 
 $^1$: \W{1}-h.~wrt.~$\nrl+\ndl$; $^2$: \W{1}-h.~wrt.~$\tau$.
 Results for tractability propagate to the left, for intractability to the right.}
 \label{fig:results}
\end{figure*}
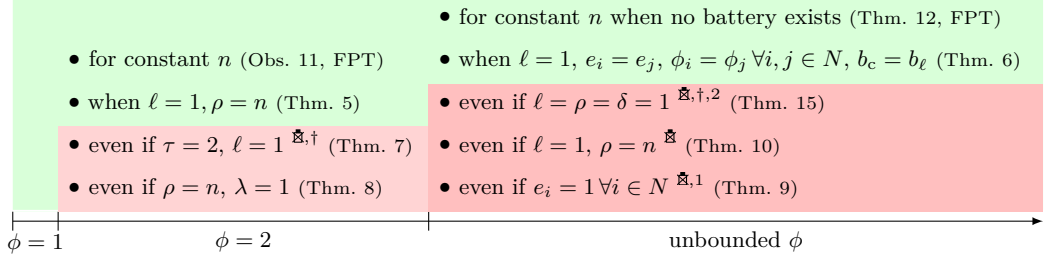

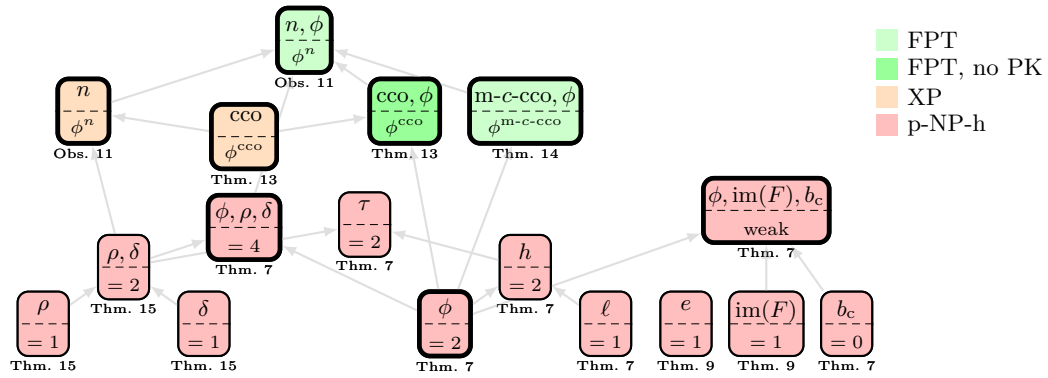
\begin{figure}[t!]
 \centering
 \begin{tikzpicture}
  \def\xr{1}\def\yr{1}
  \def\xs{2.125*\xr}\def\ys{0.75*\yr}
  \tikzstyle{xnode}=[rectangle,rounded corners,minimum width=0.675cm,minimum height=0.85cm, thick]
  \tikzstyle{xxnode}=[rectangle,rounded corners,minimum width=0.75cm,minimum height=0.5cm,ultra thin,draw]
  \def\cPnph{red!25!white}
  \def\cFpt{green!20!white}
  \def\cFptNoPK{green!40!white}
  \def\cWh{yellow!50!white}
  \def\cXp{orange!25!white}
  \def\cOpen{blue!15!white}
  \tikzstyle{nPnph}=[xnode,fill=\cPnph,draw]
  \tikzstyle{nPnphs}=[xnode,fill=\cPnph, draw]
  \tikzstyle{nFpt}=[xnode,fill=\cFpt,draw]
  \tikzstyle{nFptNoPK}=[xnode,fill=\cFptNoPK,draw]
  \tikzstyle{nWh}=[xnode,fill=\cWh,draw]
  \tikzstyle{nXp}=[xnode,fill=\cXp,draw]
  \tikzstyle{nOpen}=[xnode,fill=\cOpen,draw]
  \tikzstyle{xedge}=[-, thick, lightgray!50!white]
  
  \newcommand{\dnode}[7]{%
    \node (#1) at (#2*\xs,#3*\ys)[#4, label={[label distance=-2.25pt]-90:{\tiny#6}}]{\phantom{\scriptsize#5}};
    \draw[-,densely dashed, very thin, black] (#1.east) to (#1.west);
    \node at (#1.north)[anchor=north]{\footnotesize#5};
    \node at (#1.south)[anchor=south]{\scriptsize#7};
  }
  
  \newcommand{\thenodes}{%
    \dnode{tau}{0.5}{-0.25}{nPnph}{$\tau$}{\Cref{prop:NPhard:Part1}}{$=2$}
    \dnode{ell}{2}{-2}{nPnph}{$\mx{\ell}$}{\Cref{prop:NPhard:Part1}}{$=1$}
    \dnode{e}{2.5}{-2}{nPnph}{$\mx{e}$}{\Cref{prop:whard}}{$=1$}
    \dnode{h}{1.5}{-1}{nPnph}{$\mx{h}$}{\Cref{prop:NPhard:Part1}}{$=2$}
    \dnode{s}{1}{-2}{nPnph,line width=1.75pt,draw}{$\mx{\flex}$}{\Cref{prop:NPhard:Part1}}{$=2$}
    \dnode{F}{3}{-2}{nPnph}{$\im(F)$}{\Cref{prop:whard}}{$= 1$}
    \dnode{Bs}{3.5}{-2}{nPnph}{$\bc$}{\Cref{prop:NPhard:Part1}}{$=0$}
    \dnode{n}{-1.25}{1.75}{nXp,line width=1.75pt,draw}{$n$}{\Cref{obs:FPTnPlusFlex}}{$\mx{\flex}^n$}
    \dnode{occ}{-0.25}{1.3}{nXp,line width=1.75pt,draw}{~$\cco{}$~}{\Cref{obs:concom}}{$\mx{\flex}^{\cco}$}
    
    \dnode{R}{-1.5}{-2}{nPnph}{$\nrl$}{\Cref{prop:NPhard:BP:rsame}}{$=1$}
    \dnode{D}{-0.5}{-2}{nPnph}{$\ndl$}{\Cref{prop:NPhard:BP:rsame}}{$=1$}
    
    \dnode{RD}{-1}{-1}{nPnph}{$\nrl,\ndl$}{\Cref{prop:NPhard:BP:rsame}}{$=2$}
    \dnode{RDs}{-0.25}{-0.3}{nPnph,line width=1.75pt,draw}{$\mx{\flex},\nrl,\ndl$}{\Cref{prop:NPhard:Part1}}{$=4$}
    \dnode{ns}{0.125}{3}{nFpt,line width=1.75pt,draw}{$n,\mx{\flex}$}{\Cref{obs:FPTnPlusFlex}}{$\mx{\flex}^n$}
    \dnode{socc}{0.75}{1.75}{nFptNoPK,line width=1.75pt,draw}{$\cco,\mx{\flex}$}{\Cref{obs:concom}}{$\mx{\flex}^{\cco}$}
    \dnode{smccco}{1.5}{1.75}{nFpt,line width=1.75pt,draw}{$\mccco,\mx{\flex}$}{\Cref{prop:mccco}}{$\mx{\flex}^{\mccco}$}
    
    \dnode{sFBs}{3}{0}{nPnph,line width=1.75pt,draw}{$\mx{\flex},\im(F),\bc$}{\Cref{prop:NPhard:Part1}}{weak}

  }

  \thenodes{}
  \foreach \x/\y in {s/socc,n/ns,
    s/smccco,smccco/ns,
    ell/h,s/h,
    h/tau,
    occ/n,occ/socc,
    socc/ns,
    R/RD,D/RD,RD/n,RD/RDs,s/RDs,RDs/ns,RD/tau,
    s/sFBs,F/sFBs,Bs/sFBs%
    }
    {\draw[xedge,->,>=latex] (\x) to (\y);}
  \thenodes{}
  
  \node at (3.75*\xs, 3.0*\ys)[rectangle,fill=\cFpt,inner sep=4pt,label=0:{\small FPT}]{};
  \node at (3.75*\xs, 2.5*\ys)[rectangle,fill=\cFptNoPK,inner sep=4pt,label=0:{\small FPT, no PK}]{};
  \node at (3.75*\xs, 2.0*\ys)[rectangle,fill=\cXp,inner sep=4pt,label=0:{\small XP}]{};
  \node at (3.75*\xs, 1.5*\ys)[rectangle,fill=\cPnph,inner sep=4pt,label=0:{\small p-NP-h}]{};

 \end{tikzpicture}
 \caption{Hasse diagram of our parameters with 
 running times of algorithms or hardness specifics
 in the lower half of each cell.
 When a parameter~$p$ points to a parameter~$p'$, 
 then there is a function~$f$ such that~$p\leq f(p')$ for all instances.
 Parameters are combined additively 
 (e.g., $p,p'$ refers to~$p+p'$).
 }
 \label{fig:hasse}
\end{figure}

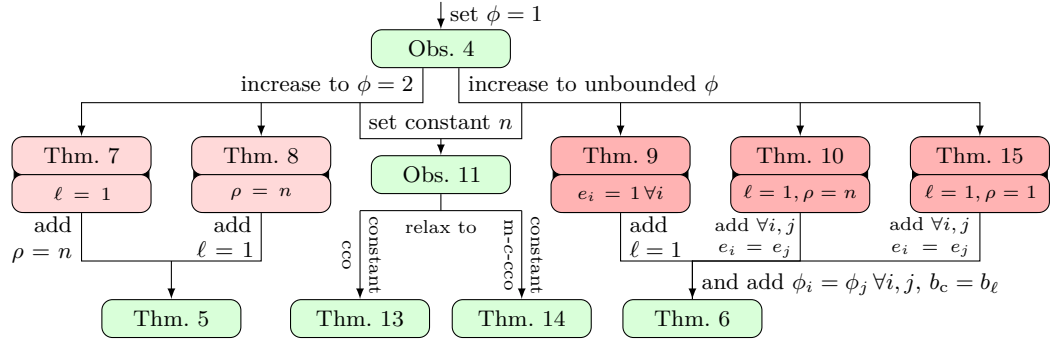
\begin{figure}[t]
 \centering
 \begin{tikzpicture}[font=\footnotesize]
  \def\xr{0.99}
  \def\yr{1}
  \def\bxw{1.8}
  \def\bxh{0.5}
  \def\xs{2.4}
  \def\ys{1.9}
  \def\cPoly{green!15!white}
  \def\cwNPh{red!15!white}
  \def\csNPh{red!30!white}

  \newcommand{\thmBoxP}[4]{%
    \node (#1) at ($(#2*\xs*\xr,#3*\ys*\yr)$)[fill=\cPoly, minimum width=1*\bxw*\xr cm, minimum height=1*\bxh*\yr cm,text width=0.9*\bxw*\xr cm, align=center,rounded corners,draw]{#4};
  }

  \newcommand{\thmBoxN}[6]{%
    \node (#1) at ($(#2*\xs*\xr,#3*\ys*\yr)$)[fill=black, minimum width=1*\bxw*\xr cm, minimum height=2*\bxh*\yr cm,rounded corners,draw]{};
    \node at (#1.north)[anchor=north, fill=#4, minimum width=1*\bxw*\xr cm, minimum height=1*\bxh*\yr cm, text width=0.9*\bxw*\xr cm, align=center, rounded corners,draw]{#5};
    \node at (#1.south)[anchor=south, fill=#4, minimum width=1*\bxw*\xr cm, minimum height=1*\bxh*\yr cm, text width=0.9*\bxw*\xr cm,align=center, rounded corners,draw,font=\scriptsize]{#6};
  }

  \thmBoxP{one}{0}{-0.125}{\Cref{obs:trivial}}
  \thmBoxP{two}{-1.5}{-2}{\Cref{prop:oneortwouniquerels}}
  \thmBoxP{fpt}{0}{-1}{\Cref{obs:FPTnPlusFlex}}
  \thmBoxP{concom}{-0.45}{-2}{\Cref{obs:concom}}
  \thmBoxP{mccco}{0.45}{-2}{\Cref{prop:mccco}}

  \thmBoxP{unbnd}{1.4}{-2}{\Cref{prop:greedy}}

  \thmBoxN{wNPhl1}{-2}{-1}{\cwNPh}{\Cref{prop:NPhard:Part1}}{$\ell=1$}
  \thmBoxN{wNPhrn}{-1}{-1}{\cwNPh}{\Cref{prop:horiz:joborders}}{$\nrl=n$}

  \thmBoxN{sNPhl1r1}{3}{-1}{\csNPh}{\Cref{prop:NPhard:BP:rsame}}{$\ell=1,\nrl=1$}
  \thmBoxN{sNPhl1rn}{2}{-1}{\csNPh}{\Cref{prop:NPhard:BP:rdiff}}{$\ell=1,\nrl=n$}
  \thmBoxN{sNPhe1}{1}{-1}{\csNPh}{\Cref{prop:whard}}{$e_i=1\,\forall i$}

  \draw[->,>=latex] ($(one)+(0,0.325*\ys*\yr)$) -- node[midway,right]{set $\phi=1$}(one);
  \draw[->,>=latex] ($(one.south)+(-0.1*\xs*\xr,0)$) -- node[midway,left]{increase to $\phi=2$}($(one)+(-0.1*\xs*\xr,-0.375*\ys*\yr)$) -- ($(wNPhl1)+(0,0.5*\ys*\yr)$) -- (wNPhl1);
  \draw[->,>=latex] ($(wNPhrn)+(0,0.5*\ys*\yr)$) -- (wNPhrn);
  \draw[->,>=latex] ($(one.south)+(0.1*\xs*\xr,0)$) -- node[midway,right]{increase to unbounded $\phi$}($(one)+(0.1*\xs*\xr,-0.375*\ys*\yr)$) -- ($(sNPhl1r1)+(0,0.5*\ys*\yr)$) -- (sNPhl1r1);
  \draw[->,>=latex] ($(sNPhl1rn)+(0,0.5*\ys*\yr)$) -- (sNPhl1rn);
  \draw[->,>=latex] ($(sNPhe1)+(0,0.5*\ys*\yr)$) -- (sNPhe1);
  \draw[-,>=latex] ($(one)+(-0.45*\xs*\xr,-0.375*\ys*\yr)$) -- ($(fpt)+(-0.45*\xs*\xr,0.25*\ys*\yr)$) -- node[midway,above]{set constant~$n$}($(fpt)+(0.45*\xs*\xr,0.25*\ys*\yr)$) -- ($(one)+(0.45*\xs*\xr,-0.375*\ys*\yr)$);
  \draw[<-,>=latex] (fpt.north) -- ($(fpt)+(0,0.25*\ys*\yr)$);

  \draw[-,>=latex] (wNPhrn.south) -- node[midway,left,align=right,text width=0.4*\xs*\xr cm]{add $\ell=1$}($(wNPhrn)+(0,-0.6*\ys*\yr)$) -- ($(wNPhl1)+(0,-0.6*\ys*\yr)$) -- node[midway,left,align=right,text width=0.4*\xs*\xr cm,yshift=-0.03*\yr cm]{add $\nrl=n$}(wNPhl1.south);
  \draw[<-,>=latex] (two.north) -- ($(two)+(0,0.4*\ys*\yr)$);

  \draw[-,>=latex] (sNPhl1r1.south) -- node[midway,left,align=right,text width=0.675*\xs*\xr cm,xshift=0.00*\xs*\xr cm,yshift=0.00*\ys*\yr cm,font=\scriptsize]{add $\forall i,j$ $e_i=e_j$}($(sNPhl1r1)+(0,-0.6*\ys*\yr)$) -- ($(unbnd)+(0,0.4*\ys*\yr)$) -- node[midway,right,align=left,text width=1.7*\xs*\xr cm,yshift=-0.03*\ys*\yr cm]{and add $\flex_i=\flex_j\,\forall i,j$, $\bc=\bl$}(unbnd.north);
  \draw[-,>=latex] (sNPhl1rn.south) -- ($(sNPhl1rn)+(0,-0.6*\ys*\yr)$) -- node[midway,above,align=center,text width=0.675*\xs*\xr cm,xshift=0.05*\xs*\xr cm,yshift=-0.04*\ys*\yr cm,font=\scriptsize]{add $\forall i,j$ $e_i=e_j$}($(unbnd)+(0,0.4*\ys*\yr)$) -- (unbnd.north);
  \draw[->,>=latex] (sNPhe1.south) -- node[midway,right,align=left,text width=0.4*\xs*\xr cm,yshift=0.0*\ys*\yr cm]{add $\ell=1$}($(sNPhe1)+(0,-0.6*\ys*\yr)$) -- ($(unbnd)+(0,0.4*\ys*\yr)$) -- (unbnd.north);

  \draw[-,>=latex] (fpt.south) -- ($(fpt)+(0,-0.25*\ys*\yr)$);
  \draw[-,>=latex] ($(concom)+(0,0.75*\ys*\yr)$) -- node[midway, sloped,below,font=\scriptsize]{relax to}($(mccco)+(0,0.75*\ys*\yr)$);
  \draw[->,>=latex] ($(mccco)+(0,0.75*\ys*\yr)$) -- node[midway, sloped,above,font=\scriptsize]{constant}node[midway, sloped,below,font=\scriptsize]{$\mccco$}(mccco.north);
  \draw[->,>=latex] ($(concom)+(0,0.75*\ys*\yr)$) -- node[midway, sloped,above,font=\scriptsize]{constant}node[midway, sloped,below,font=\scriptsize]{$\cco$}(concom.north);
 \end{tikzpicture}
 \caption{Organization chart of our structural results. Green boxes correspond to polynomial-time solvability, light-red boxes to weak \NP-hardness,
 and (darker) red boxes to strong \NP-hardness.}
 \label{fig:organigram}
\end{figure}

As to the algorithmic and experimental focus of the presentation together with the space constraints,
we defer all hardness proofs to the appendix;
we also only provide proof sketches to our algorithmic results. 
Details to results marked with (\appsymb{}) can be found in the appendix.

\section{Preliminaries}
\label{sec:basics}
We denote by~$\N$ and~$\Nzero$ the natural numbers ex- and including zero,
respectively.
We denote by~$\Qp$ and~$\Qnn$ the set of all positive and non-negative numbers from the rational numbers~$\Q$,
respectively.
For a function~$f$,
we
denote by~$\Im(f)$ the image of~$f$, and by~$\im(F)=|\Im(F)|$.

We distinguish between weak and strong \NP-hardness in the standard sense.
Weak \NP-hardness arises from numerical parameters encoded in binary and does not preclude pseudo-polynomial-time algorithms,
whereas strong \NP-hardness rules out such algorithms unless $\classP=\NP$.
This distinction is relevant in our setting since several hardness results are obtained via reductions from \prob{Partition}, \prob{Bin Packing}, or variants with unary encodings.

We use basic terminology from parameterized algorithms and complexity~\cite{CyganFKLMPPS15}.
A parameterized problem is a language~$L\subseteq \Sigma^*\times \Nzero$
over a fixed finite-sized alphabet~$\Sigma$.
$L$ is in \XP{} if there are computable functions~$f,g$ only depending on~$p$ such that every instance~$I=(x,p)$ can be decided for~$L$ in~$f(p)\cdot |I|^{g(p)}$ time.
$L$ is fixed-parameter tractable (\FPT) if $g\equiv c$ is a constant,
i.e.,
every instance~$I=(x,p)$ can be decided for~$L$ in~$f(p)\cdot |I|^c$ time;
we also say that~$L$ is in \FPT.
If $L$ is \W{1}-hard,
then it is presumably not in FPT.
Such hardness is shown via parameterized reductions,
basically translating the parameter of the first problem into the parameter of the second problem.
A kernelization is a polynomial-time algorithm that transforms each instance~$I=(x,p)$ into an decision-equivalent instance~$I'=(x',p')$ such that~$|I'|\leq f(p)$ for some computable function~$f$ only depending on~$p$;
if~$f$ is a polynomial, then we call it a polynomial kernelization (PK).

%

%
%
%
\section{Model, Problem Definition, and Parameters}
\label{sec:mpdp}

\subparagraph*{Our Model.}
For a time window~$T=\set{\tau}$ of~$\tau$ time steps,
we have an energy~forecast $F\colon T \to \Q$
(see \cref{fig:model}, accompanying this section).
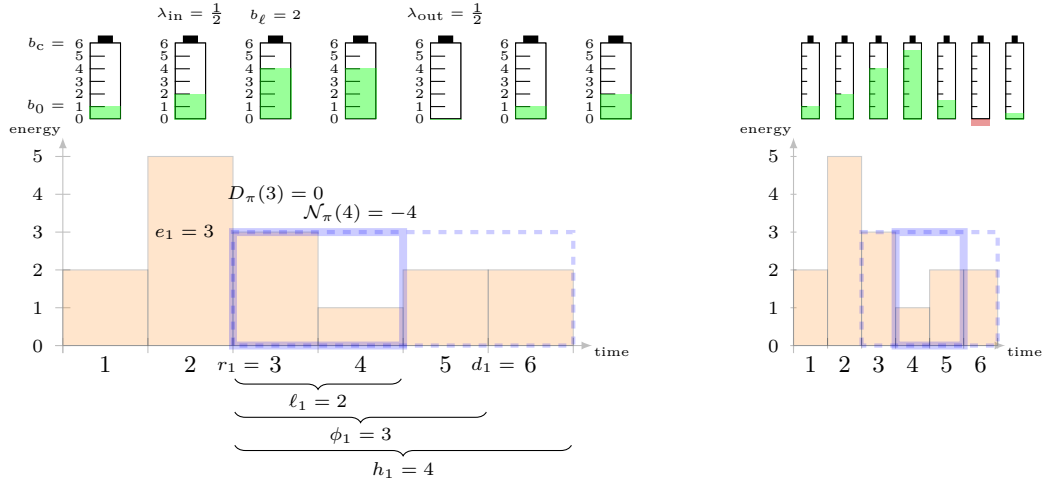
\begin{figure}[t]
 \centering
 \begin{tikzpicture}
  \def\xr{1}
  \def\yr{1}
  \def\xsh{1.125}
  \def\ysh{0.5}
  \def\bwdth{0.4}
  \def\bhght{1}
  
  \def\colF{orange}
  \def\colB{green}
  \def\colE{red!80!black}
  \def\colJ{blue}
  
  \newcommandx{\tikzBatt}[4][2=\colB]
  {
    \node (b) at (#1)[minimum width=1*\bwdth*\xr cm, minimum height=1*\bhght*\yr cm,draw,fill=white]{};
    \draw[line width=3pt] ($(b.north)+(-0.25*\bwdth*\xr,0.04*\yr)$) -- ($(b.north)+(0.25*\bwdth*\xr,0.04*\yr)$);
    \node at (b)[minimum width=1*\bwdth*\xr cm, minimum height=1*\bhght*\yr cm,draw,fill=white]{};
    \foreach \x in {0,...,6}{
      \coordinate (b\x) at ($(b.south west)+(0,\x*\bhght/6)$);
      \ifnum#4=0
        \node at (b\x)[inner sep=-1pt,label={[label distance=-2pt]180:{\tiny $\x$}}]{};
      \fi
      \ifnum\x>0
        \ifnum\x<6
          \draw (b\x) -- ($(b\x)+(0.2*\xr,0)$);
      \fi\fi
    }
    \coordinate (b55) at ($(b.south west)+(0,5.5*\bhght/6)$);
    \coordinate (b15) at ($(b.south west)+(0,1.5*\bhght/6)$);
    \coordinate (bm05) at ($(b.south west)+(0,-0.5*\bhght/6)$);
    \coordinate (b05) at ($(b.south west)+(0,0.5*\bhght/6)$);
    
    \draw[fill=#2,opacity=0.4,draw=none] (b.south east) rectangle (b#3);
    
  }
  
  \newcommand{\theCoords}{%
    
    \foreach \x in {1,...,6}{
      \node (x\x) at (\x*\xsh*\xr,-0.5*\ysh*\yr)[inner sep=0pt]{\small $\x$};
      \node (xx\x) at (\x*\xsh*\xr-0.5*\xsh*\xr,-0.0*\ysh*\yr)[inner sep=0pt]{};
      \ifnum\x=6
        \node (xx7) at (\x*\xsh*\xr+0.5*\xsh*\xr,-0.0*\ysh*\yr)[inner sep=0pt]{};
        \node (x7) at (7*\xsh*\xr,-0.0*\ysh*\yr)[inner sep=0pt]{};
      \fi
    }
    
    \foreach \y in {0,...,5}{
      \node (y\y) at (0.4*\xsh*\xr,\y*\ysh*\yr)[inner sep=0pt,label=180:{\scriptsize $\y$}]{};
    }
    \node (yBat) at (0.4*\xsh*\xr,7*\ysh*\yr)[inner sep=0pt]{};

    \draw[lightgray,thin,->,>=latex] ($(x1|-y0)+(-0.5*\xsh*\xr,-0.1*\yr)$) -- ($(x1|-y5)+(-0.5*\xsh*\xr,+0.25*\yr)$);
    \node at ($(x1|-y5)+(-0.5*\xsh*\xr,+0.25*\yr)$)[label={[label distance=-9pt]135:{\tiny energy}}]{};
    
    \foreach \x in {1,...,6}{
      \draw[lightgray,thin] ($(x\x|-y0)+(-0.5*\xsh*\xr,-0.1*\xr)$) -- ($(x\x|-y0)+(-0.5*\xsh*\xr,+0.1*\yr)$);
      \ifnum\x=6
        \draw[lightgray,thin] ($(x\x|-y0)+(+0.5*\xsh*\xr,-0.1*\xr)$) -- ($(x\x|-y0)+(+0.5*\xsh*\xr,+0.1*\yr)$);
      \fi
    }
    
    \draw[lightgray,thin,->,>=latex] ($(x1|-y0)+(-0.5*\xsh*\xr,0)+(-0.1*\xr,0)$) -- ($(x6|-y0)+(+0.5*\xsh*\xr,0)+(+0.25*\xr,0)$);
    \node at ($(x6|-y0)+(+0.5*\xsh*\xr,0)+(+0.25*\xr,0)$)[label={[label distance=-9pt]-45:{\tiny time}}]{};
    
    \foreach \y in {0,...,5}{
      \draw[lightgray,thin] ($(x1|-y\y)+(-0.5*\xsh*\xr,0)+(-0.1*\xr,0)$) -- ($(x1|-y\y)+(-0.5*\xsh*\xr,0)+(+0.1*\xr,0)$);
    }
  }
  
  \newcommand{\theF}{
    \foreach \x/\z/\y in {1/2/2,2/3/5,3/4/3,4/5/1,5/6/2,6/7/2}{
      \draw[fill=\colF,opacity=0.2] (xx\x|-y0) rectangle (xx\z|-y\y);
    }
  }
  
  \newcommand{\theJob}[1]{
    \draw[draw=\colJ,dashed,ultra thick,opacity=0.2] (xx3|-y0) rectangle (xx7|-y3);

    \ifnum#1=0
      \draw[draw=\colJ,line width=3pt,opacity=0.2] (xx3|-y0) rectangle (xx5|-y3);
      \node at (x3)[anchor=east]{\scriptsize $r_1=\phantom{d}\!\!\!$};
      \node at (x6)[anchor=east]{\scriptsize $d_1=\phantom{d}\!\!\!$};
      \node at (xx3|-y3)[label=180:{\scriptsize $e_1=3$}]{};
      
      \draw [decorate,decoration={brace,amplitude=4pt,mirror,raise=1.25em}] (xx3) -- (xx5) node[midway,yshift=-2.125em]{\scriptsize $\ell_1=2$};
    
      \draw [decorate,decoration={brace,amplitude=4pt,mirror,raise=2.5em}] (xx3) -- (xx6) node[midway,yshift=-3.375em]{\scriptsize $\flex_1=3$};
    
      \draw [decorate,decoration={brace,amplitude=4pt,mirror,raise=3.75em}] (xx3) -- (xx7) node[midway,yshift=-4.675em]{\scriptsize $h_1=4$};
    \else
      \draw[draw=\colJ,line width=3pt,opacity=0.2] (xx4|-y0) rectangle (xx6|-y3);
    \fi
  }
  
  \theCoords{}
  
  \theF{}
  
  \theJob{0}
  
  \tikzBatt{x1|-yBat}{1}{0}
  \node at (b1)[label=180:{\tiny $\bz=\ $}]{};
  \node at (b6)[label=180:{\tiny $\bc=\ $}]{};
  \tikzBatt{x2|-yBat}{2}{0}
  \node at (b.north)[label=90:{\tiny $\lin=\frac{1}{2}$}]{};
  \tikzBatt{x3|-yBat}{4}{0}
  \node at (b.north)[label=90:{\tiny $\bl=2$}]{};
  \tikzBatt{x4|-yBat}{4}{0}
  \tikzBatt{x5|-yBat}{0}{0}
  \node at (b.north)[label=90:{\tiny $\lout=\frac{1}{2}$}]{};
  \tikzBatt{x6|-yBat}{1}{0}
  \tikzBatt{x7|-yBat}{2}{0}
  
  \node at (x3|-y4)[]{\scriptsize $D_\pi(3)=0$};
  \node at (x4|-y3)[anchor=south]{\scriptsize $\Net_\pi(4)=-4$};
  
  \begin{scope}[xshift=10*\xr cm]
    \def\xr{0.4}
    \theCoords{}
  
    \theF{}
    
    \theJob{1}
    
    \tikzBatt{x1|-yBat}{1}{1}
    \tikzBatt{x2|-yBat}{2}{1}
    \tikzBatt{x3|-yBat}{4}{1}
    \tikzBatt{x4|-yBat}{55}{1}
    \tikzBatt{x5|-yBat}{15}{1}
    \tikzBatt{x6|-yBat}[\colE]{m05}{1}
    \tikzBatt{x7|-yBat}{05}{1}

  \end{scope}

 \end{tikzpicture}
 \caption{Illustration to \cref{sec:mpdp}. A simple instance with $\tau=6$,
 one job~$J_1=(r_1,d_1,\ell_1,e_1) = (3,6,2,3)$ (blue), battery~$B=(\bz,\bc,\bl,\lin,\lout)=(1,6,2,\frac{1}{2},\frac{1}{2})$ (top), and orange forecast~$F$. 
 Scheduling~$J_1$ at times~$3$ and~$4$, respectively, leads to a (left) feasible schedule~$\pi$ and a (right) infeasible one.
 }
 \label{fig:model}
\end{figure}
Additionally,
we have a set~$\Jset$ of~$n$ non-preemptive jobs,
where for each~$i\in N\ceq \set{n}$ we have a job~$J_i = (r_i, d_i, \ell_i, e_i)$~with
\begin{align*}
 r_i\in T &\colon \text{release date},\ \ d_i\in T \colon \text{deadline},
  \ \ \ell_i\in T \colon \text{job length},\ \ e_i\in \Qp \colon \text{energy per time}.
\end{align*}

\begin{remark}
  Real-world tasks may exhibit time-varying power consumption profiles.
  In our model, however, each job is assumed to have a constant energy per time over its duration.
  This abstraction can be interpreted either as replacing the original profile by its mean power 
  (thus preserving total energy) 
  or as a conservative modeling choice that assumes the profile's maximum power throughout the job (thus ensuring feasibility under fluctuations).
  \rqed
\end{remark}

Also,
we have a 
\emph{battery} $\Bat = (\bz,\bc,\bl,\lin,\lout)$ where 
\begin{align*} 
 \bz\in \Qnn &\colon \text{initial level}, \ \  
 \bc\in\Qnn \colon \text{capacity}, \ \,
 \bl\in \Qnn \colon \text{maximum loading speed}, \\
 \lin, \lout \in (0,1] &\colon \text{effective in- and output efficiencies.}
\end{align*}

Let~$(F,\Jset,\Bat)$ be an instance.
A \emph{schedule}~$\pi\colon\Jset \to T$ is
an assignment of jobs to starting times such that
$\pi(J_i)\geq r_i$ and
$\pi(J_i)+\ell_i -1 \leq d_i$.
A job~$J_j$ is \emph{active} at time~$t\in T$ if~$t\in[\pi(J_j),\pi(J_j)+\ell_j-1]$.
We denote by $\jis(t)$ the index set of the jobs active at time~$t$.
Let
\begin{align}
 D_{\pi}(t) = F(t) - E_{\pi}(t), \text{ where } E_{\pi}(t)\ceq \sum\nolimits_{i\in \jis(t)} e_i, \label{eq:Dpit}
\end{align}
denote the \emph{net} energy at time~$t\in T$ given schedule~$\pi$.
Let~$\maxz(x)\ceq \max\{0,x\}$, 
$\minz(x)\ceq\min\{0,x\}$,
and
\begin{align}
 \Net_{\pi}(t) =  \lin\cdot \maxz(D_{\pi}(t)) + \lout^{-1}\cdot \minz(D_{\pi}(t))
 \label{eq:net}
\end{align}
denote the \emph{effective} net energy relevant to the battery,
that is,
transformed by~$\lin$ when charged 
(i.e., $D_{\pi}(t)>0$)
and by~$\lout^{-1}$ when discharged
(i.e., $D_{\pi}(t)<0$).
We denote the battery state at time~$t$ given schedule~$\pi$ 
by~$\Bat_{\pi}(t)$, 
where~$\Bat_{\pi}(1)=b_0$.
The battery state in the current time step is composed of the battery level 
and the effective net energy of the previous time step,
coupled with the maximum loading speed~$\bl$,
and upper bounded by the battery's capacity.
Formally, for all~$t\in\set[2]{\tau+1}$ we have
\begin{align}
 \Bat_{\pi}(t) &= \min \left\{\bc, \Bat_{\pi}(t-1) + \min\{\bl, \Net_{\pi}(t-1)\}\right\}. \label{eq:bat:update}
\end{align}
\vskip -\lastskip
\begin{remark}
\label{rem:chargingspeed}
We neglect discharge-rate limitations of the battery. 
For modern residential battery systems, 
the maximum discharge power 
(of several kW, cf.\ \cite{ORTH2023107299}) 
exceeds the demand of household appliances
(cf.\ \cite[Table 5]{LiuYWS21})
and is therefore rarely binding in practice.
\rqed
\end{remark}
A schedule~$\pi\colon\Jset \to T$
is \emph{feasible}
if~$\Bat_{\pi}(t)\geq 0$ for all~$t\in \set{\tau+1}$.
Note that we include the auxiliary time step~$\tau+1$ 
to verify the effective net energy of the last time step~$\tau$.
\subparagraph*{Problem Definition.}
Our central decision problem is defined as follows.

\decprob{\qeTsc{} (\qeAcr)}{tdv}
{a forecast~$F$, a set~$\calJ$ of jobs, and a battery~$\Bat$}
{there is a feasible schedule}

Given a schedule~$\pi$,
one can compute each of~\eqref{eq:Dpit},
\eqref{eq:net},
and \eqref{eq:bat:update} in polynomial time,
thus:

\begin{observation}
 \label{obs:inNP}
 \qeAcr{} is contained in \NP{}.
\end{observation}

We say that the battery has no losses if~$\lin=\lout=1$.
When we say that there is no battery,
then we assume a battery with no losses and with~$\bc=0$
(note that feasibility in this case is equivalently
defined over~$D_\pi(t)\geq 0$ for all~$t\in T$).
\subparagraph*{Further Parameters.}
The parameters defined below will be central in the subsequent complexity analysis.
We denote by~$\nrl=|R|$ and~$\ndl=|D|$
the sizes of the sets of unique release times $R\ceq \bigcup_{i\in N} \{r_i\}$ and deadlines $D\ceq \bigcup_{i\in N} \{d_i\}$,
resp.
For a job~$J_i$,
$i\in N$,
we define the \emph{horizon} by $h_i \ceq  d_i-r_i+1$,
the \emph{seat}~$S_i=\{r_i,\dots,d_i-\ell_i+1\}$,
and, 
as the size of the seat,
\[ \text{the \emph{flexibility} by } \flex_i \ceq  h_i-\ell_i+1 . \]
Flexibility captures the number of different start times a job can take;
$\flex_i=1$ corresponds to a job with no scheduling flexibility.
Note that for all~$i\in N$,
$\ell_i\leq h_i$ and~$\flex_i\leq h_i$, 
but~$\ell_i$ and~$\flex_i$ are incomparable.
Recall that in this work,
we focus on the flexibility in our analysis.
We drop the subscripts to refer to the maximum over all values,
e.g.,
$\mx{\flex}\ceq \max_{i\in N} \flex_i$ denotes the maximum flexibility over all jobs.

Let~$G(\Jset)=(V,E)$ be the (undirected interval) graph with vertex set~$V=\{v_i\mid J_i\in\Jset\}$
and edge set~$E=\{\{v_i,v_j\}\mid [r_i,d_i]\cap [r_j,d_j] \neq \emptyset\}$.
We call~$G$ the \emph{job graph}.
Intuitively, 
edges in the job graph represent potential overlaps in time windows, 
which may induce dependencies for scheduling.
We consider connected components in~$G$,
i.e., 
inclusion-wise maximal vertex subsets such that in such a subset,
every two vertices are reachable from each other via a sequence of consecutively adjacent edges.
We denote by~$\cco(G)$ the order of the largest connected component of~$G$.
A \emph{modulator to $c$-$\cco$} with constant value~$c$
is a vertex set~$W\subseteq V$ such that the modified graph $G-W$,
i.e.,
when removing $W$ and all edges with an endpoint in~$W$ from~$G$,
has
$\cco$ of size at most~$c$.
By \mccco{} we denote the smallest size of a \emph{modulator to $c$-$\cco$}.
The associated problem of computing \mccco{} is also known as \prob{$c$-Component Order Connectivity} \cite{KumarL17},
which is already \NP-hard for every~$c\geq 1$~\cite{LewisY80}.

\section{Related Work} 
As usual for problems in computational sustainability,
related work stems from two perspectives: 
from a computational perspective
(here: scheduling and parameterized complexity),
and from a sustainable perspective
(here: demand-response and energy autarky).

\subparagraph*{Scheduling and Complexity.}
By now,
many scheduling problems 
have been studied 
from a parameterized complexity perspective~\cite{MnichW14}, 
revealing several interesting open problems~\cite{MnichB18}. 
Jobs whose costs co-depend on a resource—as in our setting—are comparatively rare. 
More common are budgets on the total number of jobs~\cite{NederlofS22}, 
rejection costs~\cite{HermelinSZP22}, 
or weights tied to completion times~\cite{BampisCLLMZ14}. 
A related outsourcing-motivated model considers weighted jobs that partially depend on an external resource which itself incurs a cost~\cite{BriskornDM21}; 
however, jobs have unit length, cannot overlap, and may be overdue. 
ILP
approaches to scheduling are by now well established, 
including systematic algorithmic investigations~\cite{KnopK18}.

Models without batteries include settings where energy is constantly renewed or globally limited and activities are subject to precedence constraints~\cite{BruckerDMNP99}, 
as well as more general resource-constrained scheduling frameworks studied from a complexity-theoretical viewpoint~\cite{GanianHM20}. 
Other battery-free models assume no release dates
(i.e., all jobs are available at time step~1)
together with a common deadline and a no-overlap constraint~\cite{BentertBGKN23}; 
in contrast to our work, 
several resources may be available and jobs may require more than one simultaneously. 
Further related 
is
single-machine scheduling with processing times and energy demands, 
where recharging itself may consume time steps, 
and analyze objectives such as (weighted) completion time or late jobs~\cite{YU202540}. 
Only few works combine energy harvesting with a battery: 
some restrict to unit-length, 
weighted jobs without parallel execution and forbid harvesting while processing, 
aiming to maximize 
the total weight
of feasible jobs, 
and provide weak NP-hardness, polynomial-time, and approximation results for special cases~\cite{SchieberSV23}; 
related heuristic approaches have also been explored experimentally~\cite{KumarA18}. 
An interesting extension to our model 
is so-called battery care, 
where the battery must recharge to a minimum level before reuse and may need to be discharged to a prescribed level before the next recharge~\cite{GopalakrishnanNP2022}.

\subparagraph*{Demand-Response.}

Demand-response (DR) research studies coordination in systems with distributed renewable generation such as photovoltaic (PV) units and batteries across residential~\cite{HuaETAL24,LezamaFFV20}, hotel~\cite{WamalwaIshimwe24}, community~\cite{SangareBFPP23}, and industrial~\cite{MerkertHISSS15,PravinLLW22} settings, predominantly aiming at reducing electricity costs, peak demand, or grid imports.
In residential PV-battery systems, Lezama \etal~\cite{LezamaFFV20} formulate an MILP model and solve their problem via evolutionary algorithms; 
optimality is not guaranteed due to the stochastic approach. 
Runtimes of one to three minutes are reported. 
Similarly, 
Hua \etal~\cite{HuaETAL24} introduce an energy consumption scheduler for 
interruptible and non-interruptible 
time-flexible appliances. 
Energy reduction for unscheduled against scheduled scenarios are compared,
but only for a fixed flexibility setting.
Neither varying degrees of flexibility nor runtimes are analyzed.
At larger scales, 
Wamalwa and Ishimwe~\cite{WamalwaIshimwe24} propose an MINLP model for a PV-battery-powered hotel building
to re-schedule flexible loads 
(e.g., washing machines, dishwashers, electric stoves)
taking the end-user appliance rescheduling inconvenience into account. 
For energy communities, 
Sangare \etal~\cite{SangareBFPP23} present an MILP model that is solved optimally for small 
and heuristically (potentially suboptimally) for larger instances.
Their Type B loads resemble our jobs. 
Across these PV-battery DR works~\cite{HuaETAL24,LezamaFFV20,SangareBFPP23,WamalwaIshimwe24}, flexibility is modeled as a property of loads but not systematically varied or analyzed with respect to feasibility or algorithmic behavior. 
None classifies the associated decision problems in terms of weak or strong NP-hardness, parameterized complexity, 
or tractability boundaries.
In contrast, our work treats flexibility as a central structural parameter. Rather than focusing on solver performance for fixed formulations, we characterize the complexity landscape of autarkic scheduling and analyze how increasing flexibility influences feasibility and computational difficulty.

Finally, 
we point out that there is a broader notion of 'energy flexibility' in the DR literature~\cite{HeiderRBLH21,LundLMS15,ShahidOTHDW23},
where 
flexibility is defined as the ability of an energy network to act in response to external signals, e.g., by temporal shifting of consumption or supply adaptation. 
To the best of knowledge,
these works focusing on 'energy flexibility' are more conceptual 
and neither provide complexity-theoretic analyses nor empirically study how varying flexibility levels affect external energy reductions or runtimes of solution algorithms.

\subparagraph*{Energy Autarky and PVs.}

Combining photovoltaic systems (PVs) and batteries 
on household level
is studied in the context of
cost-minimization~\cite{HeinischOGJ2019},
indicator-based self-consumption~\cite{BARZEGKARNTOVOM20201302},
emergency power supply facing blackouts~\cite{STENZEL2018165},
combination with heat pumps and thermal storage~\cite{LANGER2020115661},
or
autarky through decentralized batteries~\cite{QuernheimW2024}.
However,
to the best of our knowledge,
none of these studies investigate the use of flexible job scheduling to achieve household or community energy autarky.

\section{Polynomial-Time Solvable Cases}
\label{sec:ptime}
\appendixsection{sec:ptime}

If
every job's flexibility is one and thus no choice is left,
then \qeAcr{} is trivial.

\begin{observation}
 \label{obs:trivial}
 \qeAcr{} is linear-time solvable if~$\mx{\flex}=1$.
\end{observation}

We will see that \qeAcr{} becomes \NP-hard already for $\mx{\flex}=2$.
Thus,
tractability persists only under further restrictions,
such as the following.
\begin{theorem}[\appref{prop:oneortwouniquerels}]
 \label{prop:oneortwouniquerels}
 \qeAcr{} is polynomial-time solvable 
 if~$\mx{\flex}=2$, 
 $\nrl=n$,
 and $\mx{\ell}=1$.
\end{theorem}
As we will see,
a flexibility of two is crucial here, 
since \qeAcr{} becomes strongly \NP-hard 
for unbounded flexibilities even when $\nrl=n$
and $\mx{\ell}=1$
(\cref{prop:NPhard:BP:rdiff}).
Note that the setup of \cref{prop:oneortwouniquerels}
implies that all jobs are weakly-ordered by their starting times.
We will see that flexibility two and all jobs being weakly-ordered is not enough for tractability (\cref{prop:horiz:joborders}).

\begin{proof}[Proof sketch]
  We use dynamic programming:
  For each job~$J_j$ we have two table entries,
  corresponding to whether~$J_j$ is scheduled on~$r_j$ or~$r_j+1$,
  that stores the battery state for the next job.
  Since all release dates are distinct,
  we can now sweep ``from left to right'',
  with at most two jobs' horizons overlapping at any time.
  Thus, 
  a table entry's update only requires the battery state for the current job,
  stored for the previous job's two cases.
\end{proof}
\appendixproof{prop:oneortwouniquerels}
{
\begin{proof}
 Assume the jobs to be enumerated according to the order of their release times.
 Assume that the following~\cref{rr:jobstart} is inapplicable and hence,
 the first time step is in~$J_1$'s seat
 (which is unique due to~$r_i\neq r_j$).
 \begin{rrule}\label{rr:jobstart}
  If the first time step is not contained in any jobs' time window,
  then update the battery level~$\bz\ceq \min\{\bc, \bz + \min\{\bl,\lin \cdot F(1)\}\}$,
  $T\ceq\set{\tau-1}$,
  shift the forecast
  and 
  all job release dates and deadlines by $-1$,
  and remove the first time step.
 \end{rrule}
 We define a dynamic program
 $P_j[x]$ with~$j\in N$ and $x\in\{0,1\}$ such that
 $P_j[x]$ equals the maximum battery state at time step~$r_{j+1}$
 (let~$r_{N+1}\ceq \tau+1$)
 when all jobs with indices~$1,\dots,j$ are feasibly scheduled and job~$J_j$ is scheduled at~$r_j+x$,
 or~$-\infty$ if no feasible schedule exists for all jobs with indices~$1,\dots,j$ where job~$J_j$ is scheduled at~$r_j+x$.
 We define the auxiliary table~$P_j'[x]$ analogously to $P_j[x]$,
 but it tracks the battery state at time step~$d_j$.
 Let~$P_1'[0]$ be the battery state at time step 2 when exactly job~$J_1$ is scheduled to start at time step~$1$,
 and let $P_1'[1]$ be the battery state at time step 2 when 
 no job is scheduled at time step 1.
 For sake of readability, we
 only describe~$D(t)$ and hide all the transformation that goes into the battery update in $\langle\cdot\rangle$.
 We set
 \[
  P_j'[x] = \min\left\{
                \bc,
                \max_{y\in \{0,1\}}\left\{
                P_{j-1}[y] + \langle F(r_j)-(1-x)\cdot e_j-y\cdot e_{j-1}^*\rangle
                \right\}\right\}, 
 \]
   where
 \[
  e_{j-1}^* = \begin{cases}
               e_{j-1}, & \cif{}d_{j-1}=r_j, \\
               0, & \cotw;
              \end{cases}
 \]
 Let $P_j[x]$ be $P_j'[x]$ if~$d_j=r_{j+1}$,
 and otherwise the battery updated from~$d_j$ to~$r_{j+1}$ starting from~$P_j'[x]$,
 where we distinguish the first update regarding from which~$x$ we start:
 \[
 \Bat(d_j+1)[x] = \min\{\bc, P_j'[x] + \langle F(d_j) - x\cdot e_j \rangle\} 
 \]
 If~$P_j'[x]$ or~$\Bat(t)[x]$ for some~$t\in\set[d_j+1]{r_{j+1}}$ turns negative,
 then~$P_j[x]$ is set to $-\infty$.
 When~$P_N[x]\geq 0$ for some~$x\in\{0,1\}$,
 then we return~$\yes$.
 
 Each table has~$2\cdot N$ entries,
 and with two tables,
 each updated in polynomial time,
 filling the tables is doable in polynomial time.
 We next discuss the correctness. 
 For~$P_1[x]$,
 the correctness is clear from the definition.
 Suppose the correctness holds for job indices~$1,\dots,j-1$.
 We prove the correctness for~$P_j[x]$, $x\in\{0,1\}$.
 
 Assume there is a feasible schedule~$\pi$ for jobs~$1,\dots,j$ where~$J_j$ is scheduled at~$r_j+x$,
 $x\in\{0,1\}$,
 and the maximum battery state at time step~$r_{j+1}$ is~$b$.
 By induction,
 we know that there is~$y\in\{0,1\}$ such that~$P_{j-1}[y] \geq \Bat_{\pi}(r_j)\geq 0$.
 Since~$P'$ mimics the update of the battery by one time step with~$J_j$ being scheduled at~$r_j+x$,
 we have~$P_j'[x]\geq \Bat_{\pi}(d_j)$.
 If~$d_j=r_{j+1}$,
 we are done.
 Otherwise note that the battery updates are independent of all jobs but~$J_j$,
 which is equally considered in both updates,
 and hence this direction follows.
 
 Conversely,
 assume that~$P_j[x]\geq 0$ with~$j\in N$ and~$x\in\{0,1\}$
 (and hence,
 $P_j'[x]\geq 0$).
 Then, 
 $\max_{y\in\{0,1\}} \left\{P_{j-1}[y] + \langle F(r_j)-(1-x)\cdot e_j-y\cdot e_{j-1}^*\rangle\right\}$
 is not negative,
 implying that the $P_{j-1}[y]$ with corresponding~$y$ is not negative as well
 (recall that the only other option is~$-\infty$).
 By induction,
 we know that there is feasible schedule~$\pi'$ of the jobs with indices~$1,\dots,j-1$ such that~$J_{j-1}$ is scheduled at~$r_{j-1}+y$ and the maximum battery state is~$P_{j-1}[y]$ at time step~$r_{j}$.
 Since the battery update is deterministic for each choice of~$x$ in~$P_j[x]$,
 this direction follows.
\end{proof}
}

Finally,
we show that if all jobs are equal except for their release times,
and the battery can fully recharge in one time step,
then the problem becomes polynomial-time solvable.

\begin{theorem}[\appref{prop:greedy}]
 \label{prop:greedy}
 \qeAcr{} is polynomial-time solvable if $\mx{\ell}=1$,
 $e_i=e_j$ and
 $\flex_i=\flex_j$ for all~$i,j\in N$,
 and~$\bl=\bc$.
\end{theorem}

%
\begin{proof}[Proof sketch]
  We run the following algorithm with initially $\calJ^*=\emptyset$.
  For~$t=1,\dots,\tau$ in ascending order,
  first add all jobs with release date~$t$ to~$\calJ^*$
  and then, 
  in ascending order of their deadlines,
  try to schedule at time step~$t$ as many jobs as the forecast and battery at~$t$ allows.
  Intuitively,
  \cref{algo:greedy} is correct in this case 
  due to the following.
  On the one hand,
  the jobs are nicely orderable and as similar that we can swap them in a
  schedule to fit the order.
  On the other hand,
  since the battery can recharge from 0 to 100\% in one time step,
  there is no incentive to delay a job.
  Interestingly,
  our algorithm as such is incorrect if we drop this requirement of~$\bc=\bl$,
  as shown by the example in \cref{fig:greedy} (see appendix).
  \toappendix{
  \begin{figure}[t]
    \centering
    \begin{tikzpicture}[font=\footnotesize]
    \def\xr{1}
    \def\yr{1}
    \def\xs{0.75*\xr}
    \def\ys{0.33*\xr}
    \def\xss{4.75*\xr}
    
    \newcommand{\cntexSkel}[1]{%
      \draw[->] (0,0) to (4*\xs,0);
      \draw[->] (0,0) to (0,4.5*\ys);
      \draw[->] (4*\xs,0) to (4*\xs,4.5*\ys);
      \foreach \t in {1,2,3}{
        \draw[-] (\t*\xs,0+0.1*\yr) -- (\t*\xs,0-0.1*\yr);
        \node at (\t*\xs-0.5*\xs,0*\ys)[anchor=north]{\t};
      }
      \node at (4*\xs-0.5*\xs,0*\ys)[anchor=north]{4};
      \node at (4*\xs,0*\ys)[anchor=north]{~~$t$};
      \foreach \y in {0,1,2,3,4}{
        \draw[-] (0+0.1*\xr,\y*\ys) -- (0-0.1*\xr,\y*\ys);
        \node at (0,\y*\ys)[anchor=east]{\y};
      }
      \node at (0,4.5*\ys)[anchor=south]{$F$ \& $B$};
      \foreach \y in {1,2,3,4}{
        \draw[-] (4*\xs+0.1*\xr,\y*\ys) -- (4*\xs-0.1*\xr,\y*\ys);
        \node at (4*\xs,\y*\ys-0.5*\ys)[anchor=west]{\y};
      }
      \node at (4*\xs,4.5*\ys)[anchor=south]{Job ID};
      \draw[color=orange,ultra thick,opacity=0.5] (0,0) -- (1*\xs,0) -- (1*\xs,3*\ys) -- (2*\xs,3*\ys) -- (2*\xs,0*\ys) -- (4*\xs,0*\ys);
      \ifnum#1=0
        \def\locop{0.15}
      \else
        \def\locop{0.1}
      \fi
      \draw[draw=none,fill=yellow,opacity=\locop] (0,0) rectangle (2*\xs,1*\ys);
      \draw[draw=none,fill=green,opacity=\locop] (0,1*\ys) rectangle (2*\xs,2*\ys);
      \draw[draw=none,fill=magenta,opacity=\locop] (2*\xs,2*\ys) rectangle (4*\xs,3*\ys);
      \draw[draw=none,fill=brown,opacity=\locop] (2*\xs,3*\ys) rectangle (4*\xs,4*\ys);
      
      \ifnum#1=0
        \draw[color=blue,ultra thick,opacity=0.5] (0*\xs,1*\ys) -- (1*\xs,1*\ys);
      \fi
      \ifnum#1=1
        \draw[draw=yellow,fill=yellow,opacity=0.5] (0,0) rectangle (1*\xs,1*\ys);
        \draw[draw=green,fill=green,opacity=0.4] (1*\xs,1*\ys) rectangle (2*\xs,2*\ys);
        \draw[draw=magenta,fill=magenta,opacity=0.4] (2*\xs,2*\ys) rectangle (3*\xs,3*\ys);
        \draw[draw=brown,fill=brown,opacity=0.4,dashed] (3*\xs,3*\ys) rectangle (4*\xs,4*\ys);
        \draw[color=blue,ultra thick,opacity=0.5] (0*\xs,1*\ys) -- (1*\xs,1*\ys) -- (1*\xs,0*\ys) -- (2*\xs,0*\ys) -- (2*\xs,1*\ys) -- (3*\xs,1*\ys) -- (3*\xs,0*\ys) -- (4*\xs,0*\ys) -- (4*\xs,-1*\ys);
      \fi
      
      \ifnum#1=2
        \draw[draw=yellow,fill=yellow,opacity=0.5] (1*\xs,0) rectangle (2*\xs,1*\ys);
        \draw[draw=green,fill=green,opacity=0.4] (1*\xs,1*\ys) rectangle (2*\xs,2*\ys);
        \draw[draw=magenta,fill=magenta,opacity=0.4] (2*\xs,2*\ys) rectangle (3*\xs,3*\ys);
        \draw[draw=brown,fill=brown,opacity=0.4] (3*\xs,3*\ys) rectangle (4*\xs,4*\ys);
        
        \draw[color=blue,ultra thick,opacity=0.5] (0*\xs,1*\ys) -- (1*\xs,1*\ys) -- (2*\xs,1*\ys) -- (2*\xs,2*\ys) -- (3*\xs,2*\ys) -- (3*\xs,1*\ys) -- (4*\xs,1*\ys) -- (4*\xs,0*\ys);
      \fi
    }
    
    \newcommand{\mylbl}[1]{\node at (-1*\xs,5.375*\ys)[]{(#1)};}
    
    \begin{scope}
      \mylbl{a}
      \cntexSkel{0}
    \end{scope}
    
    \begin{scope}[xshift=1*\xss cm]
      \mylbl{b}
      \cntexSkel{1}
    \end{scope}
    
    \begin{scope}[xshift=2*\xss cm]
      \mylbl{c}
      \cntexSkel{2}
    \end{scope}

    \end{tikzpicture}
    \caption{Counter example with~$\bl=1<\bc$,
    where the orange line describes the forecast,
    the blue line the battery,
    and filled rectangles either a job's horizon (more opaque)
    or a scheduled job (less opaque).
    (a) The input instance.
    (b) An infeasible schedule output by \cref{algo:greedy}.
    (c) A feasible schedule.
    }
    \label{fig:greedy}
  \end{figure}
  }
\end{proof}
\appendixproof{prop:greedy}
{
\begin{proof}
  Consider the following heuristic given in \cref{algo:greedy}.
  \begin{myalgo}[Greedy]\label{algo:greedy}
  Greedily schedule jobs ``from left to right'' as early as possible:
  \begin{enumerate}
  \item Sort jobs by ascending deadlines and let~$\calJ^*=\emptyset$.
  \item For~$t=1,\dots,\tau$ in ascending order:
  \begin{enumerate}
    \item Add all jobs with release date~$t$ to~$\calJ^*$ .
    \item If there is a job~$J_j$ in~$\calJ^*$ with~$d_j-\ell_j+1<t$,
    return~\no.
    \item For each job~$J_j$ in~$\calJ^*$ in increasing order of deadlines (break ties lexicographically, i.e., $J_j$ before~$J_i$ if and only if~$j<i$),
    pretend to schedule job~$J_j$ starting at~$t$:
    If $\Bat_{\pi}(t') + \Net_{\pi}(t') \geq 0$ for all~$t\leq t'\leq \tau$,
    then schedule job~$J_j$ at time~$t$ and delete it from~$\calJ^*$.
  \end{enumerate}
  \item If~$\calJ^*=\emptyset$, then return~\yes,
  otherwise return~\no.
  \end{enumerate}
  \end{myalgo}

  We prove that \cref{algo:greedy} is correct,
  that is,
  $I$ is a \yes-instance if and only if \cref{algo:greedy} returns~\yes.
  If \cref{algo:greedy} schedules all jobs,
  then $I$ is a \yes-instance,
  since each job is scheduled to start within its seat 
  and the battery stays positive due to scheduling condition of~$\Bat_{\pi}(t) + \Net_{\pi}(t) \geq 0$ where checking only one time step is sufficient due to~$\ell_i=1$.

  Conversely,
  let~$I$ be a \yes-instance.
  Let~$\alpha\colon N\to N$ be an order of the jobs indices such that
  $\alpha(i) < \alpha(j)$ if and only if~$r_i< r_j$ or $r_i=r_j$ and~$i$ is lexicographically smaller than~$j$.
  We say a schedule \emph{respects~$\alpha$}
  when~$\pi(J_i) \leq \pi(J_j) \iff \alpha(i)\leq \alpha(j)$.
  Observe that,
  by definition,
  \cref{algo:greedy} respects~$\alpha$.
  Now we claim that there is a solution~$\pi$ for~$I$ that \emph{respects~$\alpha$}.
  Let~$\pi$ such that~$\pi(J_i)>\pi(J_j)$ for two distinct jobs~$J_i$ and~$J_j$ with~$\alpha(i)<\alpha(j)$.
  By the definition of~$\alpha$,
  $r_j\geq r_i$.
  Since the horizon and energy consumptions are the same,
  swapping the scheduled starting times for $J_i$ and~$J_j$ yields again a feasible schedule.
  Iteratively applying this argument gives the sought feasible schedule.

  Let~$\pi$ be a schedule respecting~$\alpha$ such that it mimics \cref{algo:greedy} as closely as possible,
  that is,
  agrees with assigning the jobs like \cref{algo:greedy} iteratively for~$t=1,\dots,\tau$
  as closely as possible.
  If they all agree,
  we know that \cref{algo:greedy} returns~\yes.
  Suppose towards a contradiction
  that no such fully-agreeing schedule~$\pi$ exist.
  Let~$t^*$ be the smallest time step where~$\pi$ disagrees with \cref{algo:greedy}.
  There are essentially three possibilities for the set~$\calJ_{t^*}$ of all jobs 
  scheduled by~$\pi$ at~$t^*$:

  \xcase{1}{The jobs sets are of the same size but include different jobs}
  This case contradicts the fact that both schedules respect~$\alpha$.

  \xcase{2}{The job set size is larger}
  This case contradicts
  the definition of \cref{algo:greedy}
  since it greedily schedules as many jobs as possible and all jobs have the same energy consumption.

  \xcase{3}{The job set size is smaller}
  We claim that we can move a job from the next closest time step to~$t^*$ while preserving feasibility.
  Let~$t'>t^*$ be the first time step when a job is scheduled by~$\pi$
  after~$t^*$.
  Since~$\pi$ respects~$\alpha$,
  the smallest~$j$ according to~$\alpha$ for a job~$J_j$ 
  not scheduled in~$t^*$ must be scheduled in~$t'$.
  Let $\pi'$ be schedule~$\pi$ where only job~$J_j$ is rescheduled to start at~$t^*$.
  We claim that~$\pi'$ is feasible.
  Note that~$J_j$ can be scheduled by \cref{algo:greedy},
  and thus its release date is at most~$t^*$.
  Moreover,
  for all~$t^*<t''\leq t'$,
  we have
  $\Bat_{\pi'}(t'') \geq 0$ by~\cref{algo:greedy}
  and
  \begin{align}
     \Bat_{\pi}(t'') \leq  \Bat_{\pi'}(t'') + \lin\cdot e_j. \label{eq:greedy:ub}
  \end{align}
  It suffices to show that $\Bat_{\pi'}(t'+1) \geq \Bat_{\pi}(t'+1)$.
  If~$\Bat_{\pi}(t'+1)=\bc$, 
  then there is nothing to show.
  We distinguish whether $F(t') - \sum_{i\in \jis(t')} e_i$ is non-negative or negative.
  Recall that~$\bl=\bc$.

  \xcase{(a)}{$F(t') - \sum_{i\in \jis(t)} e_i\geq 0$}
  Then we have
  \begin{align*}
    \Bat_{\pi'}(t'+1) 
    &
    = \Bat_{\pi'}(t') + \lin\cdot (F(t') - \sum_{i\in \jis(t')\setminus\{j\}} e_i)
    \\
    &
    = \Bat_{\pi'}(t') + \lin\cdot e_j + \lin\cdot (F(t') - \sum_{i\in \jis(t')\setminus\{j\}} e_i - e_j)
    \\
    &
    \geq \Bat_{\pi}(t') + \lin\cdot(F(t') - \sum_{i\in \jis(t')\setminus\{j\}} e_i - e_j)
    \geq 
    \Bat_{\pi}(t'+1).
  \end{align*}
  This contradicts the choice of~$\pi$.

  \xcase{(b)}{$F(t') - \sum_{i\in \jis(t)} e_i< 0$}
  We further distinguish whether $F(t') - \sum_{i\in \jis(t)\setminus\{j\}} e_i$ is non-negative or negative.
  If $F(t') - \sum_{i\in \jis(t')\setminus\{j\}} e_i<0$,
  then
  \begin{align*}
    \Bat_{\pi'}(t'+1)
    &
    = \Bat_{\pi'}(t') + \lout^{-1}\cdot (F(t') - \sum_{i\in \jis(t')\setminus\{j\}} e_i)
    \\
    &
    = \Bat_{\pi'}(t') + \lout^{-1}\cdot e_j + \lout^{-1}\cdot(F(t') - \sum_{i\in \jis(t')\setminus\{j\}} e_i - e_j)
    \\
    &
    \geq \Bat_{\pi'}(t') + \lin\cdot e_j + \lout^{-1}\cdot(F(t') - \sum_{i\in \jis(t')\setminus\{j\}} e_i - e_j)
    \\
    &
    \geq \Bat_{\pi}(t') + \lout^{-1}\cdot (F(t') - \sum_{i\in \jis(t')\setminus\{j\}} e_i - e_j)
    \geq
    \Bat_{\pi}(t'+1).
  \end{align*}
  If $F(t') - \sum_{i\in \jis(t')\setminus\{j\}} e_i\geq 0$,
  then
  \begin{align*}
    \Bat_{\pi'}(t'+1)
    &
    = \Bat_{\pi'}(t') + \lin\cdot (F(t') - \sum_{i\in \jis(t')\setminus\{j\}} e_i)
    \\
    &
    = \Bat_{\pi'}(t') + \lin\cdot e_j + \lin\cdot(F(t') - \sum_{i\in \jis(t')\setminus\{j\}} e_i - e_j)
    \\
    &
    \geq \Bat_{\pi'}(t') + \lout^{-1}\cdot(F(t') - \sum_{i\in \jis(t')\setminus\{j\}} e_i - e_j)
    \\
    &
    \geq
    \Bat_{\pi}(t'+1).
  \end{align*}
  Either case contradicts the choice of~$\pi$.
\end{proof}
}

\section{Hardness Results}
\label{sec:hardness}
\appendixsection{sec:hardness}

We next show that a flexibility of two,
as opposed to one,
already makes the problem computationally hard.
Herein,
we observe a hierarchy:
for a flexibility of two, 
we obtain weak \NP-hardness,
and for larger (unbounded) flexibilities,
we obtain strong \NP-hardness.
\subsection{Weak Hardness for Flexibility Two}

We show that \qeAcr{} is weakly 
\NP-hard even for flexibility two.
We give two hardness reductions,
both from the 
well-known (weakly) \NP-hard~\cite{Kar72} \prob{Partition} problem
(see \cref{app:sec:hardness}).
\toappendix{
\decprob{\partTsc{}}{partition}
{a multiset~$X=\{x_1,\dots,x_n\}$ of numbers from~$\N$}
{there are two disjoint subsets~$X_1,X_2$ of~$X$ such that~$X=X_1\cup X_2$ 
and~$\sum_{x_i\in X_1} x_i= \sum_{x_j\in X_2} x_j$}
}
We first show that \qeAcr{} is already hard for two time steps.
Intuitively,
we can model the numbers in any \textsc{Partition} instance with the jobs' energies and the task to partition with 
flexibility two and length one 
for each job 
facing only two time steps overall.

\begin{theorem}[\appref{prop:NPhard:Part1}]
 \label{prop:NPhard:Part1}
 \qeAcr{} without battery 
 is weakly \NP-hard even if~$\flex=\tau=2$,
 $\mx{\ell}=1$,
 $\nrl=\ndl=1$,
 and~$F(1)=F(2)$.
\end{theorem}

\appendixproof{prop:NPhard:Part1}
{
\begin{proof}
 Let~$I=(X=\{x_1,\dots,x_n\})$ be an instance of \textsc{Partition}
 and let~$\sigma=\sum_{i=1}^n x_i$.
 We construct an instance~$I'$ of~\qeAcr{} without battery 
 and with~$\tau=2$.
 For each~$i\in N$,
 add a job~$J_i$ with $r_i = 1$, $d_i = \tau$, $\ell_i =1$, 
 and $e_i = x_i$.
 The forecast is $ F(t) = \sigma/2$ for each~$t\in\{1,2\}$.
 We claim that~$I$ is a \yes-instance
 if and only if~$I'$ is a \yes-instance.
 The correctness follows directly from our convention for schedule~$\pi$ and solution~$(X_1,X_2)$
 that $\pi(J_j)=t \iff x_j\in X_t$ for every~$j\in N$.
 Thus,
 all jobs are scheduled if and only if $(X_1,X_2)$ is a partition of~$X$.
 Moreover,
 $F(t) - \sum_{j\in \jis(t)} e_j \geq 0 \iff \sum_{x_j\in X_t} x_j \leq \sigma/2$.
\end{proof}
}

\noindent
Even if all jobs are weakly~ordered by their starting times,
\qeAcr{} is computationally challenging.

\begin{theorem}[\appref{prop:horiz:joborders}]
 \label{prop:horiz:joborders}
 \qeAcr{} is weakly \NP-hard even if $\mx{\flex}=2$, 
 $\rho=n$,
 and $\lin=\lout=1$.
\end{theorem}

\noindent
Note that $\mx{\flex}=2$ and 
$\rho=n$
imply that all jobs are weakly ordered by their starting times.
Recall that by \cref{prop:oneortwouniquerels},
we know that
when additionally~$\ell=1$,
\qeAcr{} is polynomial-time solvable.
Interestingly, 
this is the only reduction in our work that makes use of the battery.

\appendixproof{prop:horiz:joborders}
{
\begin{proof}
 Let~$I=(X = \{x_1, \dots, x_n\})$ be an instance of \textsc{Partition},
 and let $\sigma=\sum_{i=1}^n x_i$.
 We construct an instance~$I'$ of \qeAcr{} as follows.
 Let~$\tau=n+2$.
 For each~$x_i\in X$,
 construct job~$J_i$ with~$r_i=i$, $d_i=\tau$,
 $\ell_i=h_i-1$, and~$e_i=x_i$.
 The battery is~$\Bat=(0,\sigma/2,\sigma/2,1,1)$.
 The forecast is
 $F(t)= \sum_{i=1}^{t} x_i$ if $t\in N$,
 and $F(t)=\sigma/2$ if~$t\in\{n+1,n+2\}$.
 This finishes the construction.
 Note that~$J_i$ is scheduled either at time step~$i$ or~$i+1$,
 and hence, all jobs are weakly-ordered.\footnote{Note that with $2n+2$ time steps,
 the construction can be adjusted to enforce
 a strict job order.}%
 We claim that~$I$ is a \yes-instance if and only if~$I'$ is a \yes-instance.
 
 Intuitively,
 the correctness stems from the following.
 At time step~$n+1$,
 by construction,
 all jobs are scheduled to run.
 Hence,
 $\sigma$ energy is required.
 Since only~$\sigma/2$ energy is provided by the forecast,
 the battery must be at full state,
 i.e.,
 at~$\sigma/2$.
 Then,
 in time step~$n+2$,
 the battery is empty and again only~$\sigma/2$ energy is provided by the forecast.
 Now,
 the key insight is:
 If there is less than~$\sigma/2$ energy consumed in the last time step,
 the battery is also only charged with less than~$\sigma/2$
 in the first~$n$ time steps;
 a contradiction.

 \RD{}
 Let~$(X_1,X_2)$ be a solution to~$I$.
 For each~$i\in N$,
 set~$\pi(J_i)=i$ if~$x_i\in X_1$,
 and $\pi(J_i)=i+1$ if $x_i\in X_2$.
 Since~$(X_1,X_2)$ is a partition,
 all jobs are scheduled.
 Now,
 for the battery, we have that for~$i\in N$ it holds that ($\indic_A\in\{0,1\}$ denotes the indicator function that evaluates to 1 if and only if statement~$A$ is true)
 \begin{align} 
  \begin{aligned}
    B(i) 
    &= B(i-1)+F(i)-\sum_{j\in \jis(i)} x_j 
    =  B(i-1)+F(i)-\sum_{j=1}^{i-1} x_j - \indic_{x_j\in X_1} \cdot x_j
    \\
    &=  B(i-1) + x_i - \indic_{x_i\in X_1} \cdot x_i = B(i-1) + \indic_{x_i\in X_2} \cdot x_i
    \label{eq:prop:horiz:joborders}
  \end{aligned}
 \end{align}
 Thus,
 at time step~$n+1$, 
 we have that~$B(n+1) = \sum_{x_i\in X_2} x_i = \sigma/2$.
 Since $d_i=\tau$ and $\ell_i=h_i-1$,
 we know that~$B(n+2) = B(n+1) + F(n+1) - \sigma = 0$.
 Finally,
 we have that
 $B(n+3) = B(n+2) + F(n+2) - \sum_{j\in \jis(n+2)} x_j = \sigma/2 - \sum_{x_j\in X_2} x_j = 0$.
 
 \LD{}
 Let~$\pi$ be a solution to~$I'$.
 Since $d_i=\tau$ and $\ell_i=h_i-1$,
 we know that~$B(n+2) = B(n+1) + F(n+1) - \sigma = 0$,
 and thus~$B(n+1)=\sigma/2$.
 Let~$X_2 \ceq \{x_j\in X\mid j\in \jis(n+2)\}$
 and~$X_1\ceq X\setminus X_2$.
 We claim that~$(X_1,X_2)$ is a solution to~$I$.
 Since $B(n+3)\geq 0$,
 we have that~$\sum_{x_j\in X_2} x_j\leq \sigma/2$.
 Due to~\eqref{eq:prop:horiz:joborders},
 we know that~$\sigma/2=B(n+1)\leq \sum_{x_j\in X_2} x_j$.
 It follows that~$\sum_{x_j\in X_2} x_j=\sigma/2$.
\end{proof}
}

\subsection{Strong Hardness for Larger Flexibilities}

The distinct reductions behind \cref{prop:whard,prop:NPhard:BP:rdiff,prop:NPhard:BP:rsame,prop:NOPK:cco} 
use the same well-known \NP- and \W{1}-hard \cite{GareyJ79,JANSEN201339} problem \ubpTsc{} (see \cref{app:sec:hardness}).
Similar for each reduction is that,
intuitively,
each bin is represented as some distinct time interval.
\begin{theorem}[\appref{prop:whard}]
 \label{prop:whard}
 \qeAcr{} without battery is strongly \NP-hard
 even if~$e_i=1$ for all~$i\in N$ and $\Im(F)=\{1\}$.
 In this case,
 the problem is also
 \W{1}-hard \wpb{}~$\nrl+\ndl$.
\end{theorem}
Note that 
\cref{prop:whard}'s setup
implies that only one job can be scheduled at any time step.
\toappendix{
We reduce from the following \NP-hard \cite{GareyJ79} problem \ubpTsc{},
which is \W{1}-hard~\cite{JANSEN201339} \wpb{} the number of bins.

\decprob{\ubpTsc{} (\ubpAcr{})}{bp}
{a multiset~$X=\{x_1,\dots,x_n\}$ of numbers from~$\N$ and two integers~$\nbin,\sbin\in \N$,
all numbers encoded in unary}
{there are $k$ pairwise disjoint subsets~$X_1,\dots,X_\nbin$ of~$X$ such that $X=\bigcup_{i=1}^{\nbin} X_i$ and
for each~$i\in\set{\nbin}$ it holds that~$\sum_{x\in X_i} x\leq \sbin$}
}
\appendixproof{prop:whard}
 {
\begin{proof}
 Let~$I=(X=\{x_1,\dots,x_n\},\nbin,\sbin)$ be an instance of \ubpTsc{},
 construct an instance~$I'$ of \qeAcr{} without battery as follows.
 Let~$T=\set{\tau}$ with~$\tau\ceq \nbin\cdot\sbin + \nbin$.
 For each item~$x_i\in X$,
 add a job~$J_i$ with release date~$1$,
 due date~$\nbin\cdot \sbin + \nbin$,
 $e_i=1$,
 and
 length~$\ell_i=x_i$.
 Additionally,
 for each~$i\in\set{\nbin-1}$,
 add \emph{boundary} job~$J_{n+i}$ with release date~$i\cdot (\sbin + 1)$,
 due date~$i\cdot (\sbin + 1)$,
 $e_{n+i}=1$,
 and
 length~$\ell_{n+i}=1$.
 This finishes the construction.
 We claim that~$I$ is a \yes-instance if and only if~$I'$ is a \yes-instance.
 
 \RD{}
 Let~$(X_1,\dots,X_\nbin)$ be a solution to~$I$.
 For each~$i\in \set{\nbin}$,
 let~$X_i=\{x_{i_1},\dots,x_{i_{n_i}}\}$ be the numbers in~$X_i$,
 where~$n_i=|X_i|$.
 We construct a schedule~$\pi$ for~$I'$ as follows.
 First,
 schedule each boundary job in its only possible way,
 i.e.,
 for each~$i\in\set{\nbin-1}$,
 $\pi(J_{n+i})=i\cdot (\sbin+1)$.
 For each~$i\in \set{\nbin}$ and for each~$j\in \set{n_i}$,
 schedule job~$J_{i_j}$ at~$(i-1)\cdot (b + 1) + 1 + \sum_{j'=1}^{j-1} \ell_{i_{j'}}$,
 i.e.,
 schedule job~$J_{i_j}$ right after job~$J_{i_{j-1}}$ finishes.
 By this,
 no two of these jobs overlap.
 Moreover,
 since~$\sum_{j=1}^{n_i} \ell_{i_{j}} \leq \sbin$,
 these jobs are active between time steps~$(i-1)\cdot (b + 1)$
 and~$i\cdot (b+1)$.
 Thus,
 they do not overlap with the boundary jobs~$J_{n+i-1}$
 (when~$i>1$) and~$J_{n+i}$ (when~$i<k-1$),
 and the last job finishes latest at time step~$\tau$.
 
 \LD{}
 Let~$\pi$ be a feasible schedule of all jobs forming a solution to~$I'$.
 We know that each boundary job is scheduled exactly at its release date.
 Thus,
 we can distribute the jobs in the following way:
 let~$\calJ_i$ denote the set of all jobs~$J_{i_1},\dots,J_{i_{n_i}}$ that are scheduled to start at some time step in~$\set[(i-1)\cdot (\sbin + 1) + 1]{i\cdot (\sbin + 1) - 1}$,
 where~$n_i = |\calJ_i|$.
 We claim that~$(X_1,\dots,X_\nbin)$
 with~$X_i = \{x_{i_1},\dots,x_{i_{n_i}}\}$ is a solution to~$I$.
 Since every two jobs are disjoint
 and~$\sum_{j=1}^{n_i} x_{i_j} = \sum_{j=1}^{n_i} \ell_{i_j} \leq \sbin$,
 the claim follows.
\end{proof}
}
Next we show that for unit lengths
and distinct release dates,
hardness remains.
\begin{theorem}[\appref{prop:NPhard:BP:rdiff}]
 \label{prop:NPhard:BP:rdiff}
 \qeAcr{} without battery is strongly \NP-hard,
 even if $\mx{\ell}=1$,
 $\rho=n$ and~$\im(F)=2$.
\end{theorem}
Recall that when adding a flexibility of two
to the restrictions of \cref{prop:NPhard:BP:rdiff},
we arrive in a polynomial-time solvable case (\cref{prop:oneortwouniquerels}).

\appendixproof{prop:NPhard:BP:rdiff}
{
\begin{proof}
 Let~$I=(X = \{x_1, \dots, x_n\},\nbin,\sbin)$ be an instance of \textsc{Unary Bin Packing}.
 Construct an instance~$I'$ of \qeAcr{} without battery as follows.
 Let $\tau = n + \nbin$.
 For each $i\in N$, 
 add job $J_i$ with $r_i = i$, $d_i = \tau$, $\ell_i =1$, 
 and $e_i = x_i$.
 The forecast is 
 $F(t) = 0$ if $t\in N$,
 and~$F(t) = \sbin$ if $t\in\set[n+1]{n+k}$.
 Thus, 
 in every solution,
 no job is scheduled before time step~$n+1$.
 We claim that~$I$ is a \yes-instance if and only $I'$ is a \yes-instance.
 The correctness follows from the 
 convention~$\pi(J_j)=n+t \iff x_j\in X_t$ for schedule~$\pi$ and solution~$(X_1,\dots,X_{\nbin})$.
 Since then,
 all jobs are scheduled if and only if~$(X_1,\dots,X_{\nbin})$ is a partition.
 Moreover,
 for every $t\in\set[n+1]{n+k}$,
 we have~$F(t) - \sum_{j\in \jis(t)} e_j \geq 0 \iff \sum_{x_j\in X_t} x_j\leq \sbin$.
\end{proof}
}

\section{Parameterized Complexity}
\label{sec:pc}
\appendixsection{sec:pc}

\subsection{FPT results}
\label{sec:fpt}

In~\cref{sec:hardness},
we showed that \qeAcr{} is \NP-hard
for constant values of almost all natural parameters,
except for the number~$n$ of jobs. 
In fact, 
for any constant number of jobs,
the problem becomes polynomial-time solvable
since we can guess for each job its starting time.

\begin{observation}[\appref{obs:FPTnPlusFlex}]
 \label{obs:FPTnPlusFlex}
 Let~$I$ be any instance of \qeAcr{} with~$n$ jobs and maximum flexibility~$\mx{\flex}$.
 In
 $O(\flex^n \cdot \poly(|I|))$ time,
 we can either correctly report that~$I$ is a \no-instance 
 or output a feasible schedule that maximizes the battery state when the last job finishes.
\end{observation}
\appendixproof{obs:FPTnPlusFlex}
{
\begin{proof}
  We can guess the starting time of each job.
  Since each starting time must lie in the job's seat,
  we have~$\flex_j$ possibilities for job~$J_j$.
  Formally,
  for an instance~$I$ with job set~$\Jset$ consisting of~$n$ jobs,
  let~$\Pi(I)=\bigtimes_{i=1}^n S_i$ denote the set of all possible schedules,
  where~$S_i$ denotes the seat of job~$J_i$.
  Note that~$|\Pi(I))|=\prod_{i=1}^n \flex_i\leq \mx{\flex}^{n}$.
  Thus,
  we test each schedule~$\pi\in \Pi(I)$
  in polynomial time due to~\cref{obs:inNP}.
  If any schedule turns out to be feasible,
  then we return \yes,
  and \no{} otherwise.
  Since we test every schedule,
  if there is a feasible schedule,
  we can output one that maximizes the state of the battery when the last job finished.
  This results in~$O(\flex^n \cdot \poly(|I|))$ time to solve instance~$I$.
\end{proof}
}
\cref{obs:FPTnPlusFlex} leads to two natural next questions:
(Q1) Is \qeAcr{} in \FPT{} \wpb{}~$n$ alone?
(Q2) What (substantially) smaller parameters than~$n$
still lead to \FPT{}, when combined with~$\mx{\flex}$,
and \XP{}?
Regarding~(Q1),
so far, we only know that this is the case for constant flexibility and if there is no battery.
\begin{theorem}[\appref{prop:nobattery:FPTn}]
 \label{prop:nobattery:FPTn}
 \qeAcr{} without battery is \FPT{} \wpb{} $n$.
\end{theorem}
\begin{proof}[Proof sketch]
We guess the (weak) order of jobs regarding their starting times
and then iteratively along this order 
schedule each job as early as possible.
This is correct since intuitively,
the absence of a battery makes delays useless.
Recall that 
\qeAcr{} with battery
is \NP-hard (\cref{prop:horiz:joborders})
even when such a weak order is given.
\end{proof}
\appendixproof{prop:nobattery:FPTn}
{
\begin{proof}
 We guess a weak order~$\chi\colon N\to N$ on the jobs with respect to their starting times,
 where~$\chi(i)$ denotes the~$i$-th job in the ordering.
 We claim that iteratively with~$i=1,\dots,n$ 
 scheduling job~$\chi(i)$ as early as possible 
 yields a feasible schedule~$\hat{\pi}$.
 Intuitively, 
 advancing a job cannot reduce feasibility since energy cannot be stored.
 Suppose the claim is false,
 and let~$\pi$ be a feasible schedule such that
 $\pi(\chi(i')) = \hat{\pi}(\chi(i'))$ for all~$1\leq i'<i$ for some~$i<n$
 with $\pi(\chi(i)) \neq \hat{\pi}(\chi(i))$
 and~$i$ is maximal among all such schedules.
 Let~$\pi'$ be the schedule that agrees with~$\pi$ on all but~$J_j\ceq \chi(i)$ and that starts~$J_j$ earliest possible.
 We have that~$D_{\pi}(t)\geq 0$ for all~$\pi'(J_j)\leq t<\min\{\pi'(J_j)+\ell_j,\pi(J_j)\}$ by definition
 of being ``possible''.
 Moreover,
 we have
 $D_{\pi}(t) = D_{\pi'}(t)$ if $t<\pi'(J_j)$ or $t>\pi(J_j) + \ell_j$ 
 since~$\pi$ and~$\pi'$ coincide and there is no battery.
 Also, 
 we have $D_{\pi}(t) = D_{\pi'}(t)$ if $t\in \{\pi(J_j),\pi'(J_j) + \ell_j\}$
 since there is no battery and~$J_j$ is scheduled either way.
 Finally,
 $D_{\pi}(t) = D_{\pi'}(t) - e_j \leq D_{\pi'}(t)$ if $\pi'(J_j)+\ell_i\leq t\leq \pi(J_j) + \ell_j$
 since there is no battery.
 Altogether,
 the schedule is feasible and all jobs~$\chi(1),\dots,\chi(i)$ are iteratively scheduled earliest possible, 
 a contradiction to the choice of~$\pi$.
\end{proof}
}
Regarding~(Q2),
we identify two (incomparable) parameters for which this is the case.
The first one is the size~$\cco$ of the largest connected component of job graph~$G$,
and the second one is the size~\mccco{} of a smallest modulator to~$c$-$\cco$ in~$G$.
\begin{theorem}[\appref{obs:concom}]
 \label{obs:concom}
 \qeAcr{} is in \FPT{} \wpb{} $\cco$ combined with the flexibility
 and in~\XP{} \wpb{}~$\cco$ alone.
\end{theorem}
\begin{proof}[Proof sketch]
 The idea is to solve each instance induced by a connected component from left to right via the algorithm from~\cref{obs:FPTnPlusFlex}, maximizing the battery level when the last job of the component finishes. Intuitively, this is correct since the only influence between components is the battery state; thus, we pass on a maximally charged battery.
\end{proof}
\appendixproof{obs:concom}
{
\begin{proof}
 Let~$I=(F,\Jset,\Bat=(\bz,\bc,\bl,\lin,\lout))$ be an input instance to~\qeAcr{}.
 Let~$G=G(\Jset)$ be the job graph and~$C_1,\dots,C_p$ its connected components
 (ordered from left to right by time).
 For each component~$C_i$,
 let~$\Jset_i$ be the corresponding subset of jobs and
 let~$t_i^*$ be the last time step contained in a horizon of a job in~$\Jset_i$.
 We start with the instance~$I_1$ with jobset~$\Jset_1$ and battery~$\Bat=(\bz,\bc,\bl,\lin,\lout)$
 and forecast~$F_1\colon \set{t_1^*}\to\Q$ with~$F_1(t)=F(t)$.
 Using the algorithm from~\cref{obs:FPTnPlusFlex},
 we compute a feasible schedule~$\pi_1$ that maximizes~$\Bat_{\pi_1}(t_1^*+1)$
 (which is possible since the algorithm branches over all possible schedules).
 Let~$B_1^*$ be the maximum obtained value of $\Bat_{\pi_1}(t_1^*+1)$.
 Then, 
 we next consider instance~$I_2$ with jobset~$\Jset_2$ battery~$\Bat=(\bz=B_1^*,\bc,\bl,\lin,\lout)$
 and forecast~$F_2\colon \set{t_2^*}\to\Q$ with~$F_2(t)=F(t_1^*+t)$.
 Again compute a feasible schedule~$\pi_2$ that maximizes~$\Bat_{\pi_2}(t_2^*+1)$,
 and proceed with~$I_3$ defined in the same way,
 and so on.
 We claim that~$I$ is a \yes-instance if and only if
 each of~$I_1,\dots,I_p$ is a \yes-instance.
 If for each instance~$I_q$, 
 a schedule~$\pi_q$ was found,
 then
 the schedule~$\pi$ that combines~$\pi_1,\dots,\pi_p$ 
 is a schedule of~$I$
 (this is true since at each instance's border, 
 we transmit the last battery state via~$\bz$).
 
 It remains to show the other direction.
 Let~$\pi$ be a feasible schedule.
 Let~$\pi_q$ be~$\pi$ restricted to the job set~$\calJ_q$.
 If~$\pi_q$ is so that~$\Bat_{\pi_q}(t_q^*+1)$ is not maximal over all feasible schedules
 on~$t_{q-1}^*+1,\dots,t_q^*$,
 with~$t_0^*\ceq 1$,
 then we claim that we can replace~$\pi_q$ by a feasible schedule~$\pi_q'$ with maximum~$\Bat_{\pi_q'}(t_q^*+1)$ in~$\pi$ to obtain again a feasible schedule~$\pi'$.
 Doing this iteratively from ``left to right'' then proves this direction.
 By construction, 
 we know that~$\Bat_{\pi'}(t)\geq 0$ for all~$t\in\set{t_q^*}$.
 Moreover,
 it also holds true that~$\Bat_{\pi'}(t)\geq \Bat_{\pi}(t)$ for all~$t\geq t_q^*+1$.
 This is true for~$t=t_q^*+1$ by the choice of~$\pi_q'$.
 Suppose this holds true for all~$t_q^*+1\leq t<\tau$.
 We show that it also holds true for~$t+1$:
 since the schedules other than~$\pi_q$ did not change,
 we have that~$\Net_{\pi'}(t) = \Net_{\pi}(t)$.
 Thus,
 $\Bat_{\pi'}(t+1) = \min \left\{\bc, \Bat_{\pi'}(t) + \min\{\bl, \Net_{\pi'}(t)\}\right\}
 \geq \min \left\{\bc, \Bat_{\pi}(t) + \min\{\bl, \Net_{\pi}(t)\}\right\} = \Bat_{\pi}(t+1)$.
\end{proof}
}
With~\cref{obs:concom} at hand,
we obtain the following for a modulator to $c$-$\cco$.
\begin{theorem}[\appref{prop:mccco}]
 \label{prop:mccco}
 Let~$I$ be an instance of \qeAcr{},
 $G$ its job graph,
 and~$W\subseteq V$ be a modulator to~$c$-$\cco$.
 Then,
 $I$ is solvable in~$O(\phi^{|W|+c}\cdot \poly(|I|))$ time.
\end{theorem}
\appendixproof{prop:mccco}
{
\begin{proof}
 Let~$I_W=\{i\in N\mid v_i\in W\}$ be the index set corresponding to~$W$.
 For each possible partial schedule~$\pi'$ of the jobs 
 in~$\Jset_{W}\ceq \{J_i\mid i\in I_W\}$
 (of which there are~$\flex^{|\Jset_{W}|}$),
 adjust the forecasts according to the guessed schedule
 and delete all jobs in~$\Jset_{W}$.
 Let~$I_{\pi'}$ denote the resulting instance
 and~$G_{\pi'}$ denote its job graph.
 By construction,
 we have that~$G_{\pi'}$ has $\cco\leq c$.
 Thus,
 we use~\cref{obs:concom}
 to solve~$I_{\pi'}$ in $O(\mx{\flex}^{c}\cdot \poly(|I_{\pi'}|))$ time.
 Return \yes{} if and only if there is a partial schedule $\pi'$ of the jobs in~$\Jset_{W}$
 such that $I_{\pi'}$ is a \yes-instance.
 We now prove correctness with the following argument.
 Let~$F_{\pi'}$ be the forecast adjusted in accordance with~$\pi'$,
 and let~$\pi''$ be a solution to~$I_{\pi'}$.
 Let~$\pi$ be the schedule that executes all jobs in~$\Jset_W$ according to~$\pi'$ and
 all remaining jobs according to~$\pi''$.
 We have that~$F_{\pi'}(t) = F(t) - \sum\nolimits_{i\in \widehat{\pi'}(t)} e_i$ for every~$t\in T$.
 Hence,
 for every~$t$,
 we have 
 \begin{align*}
  D_{\pi}(t) &= F(t) - \sum_{i\in \jis(t)} e_i = F(t) - \sum_{i\in \widehat{\pi'}(t)} e_i - \sum_{j\in \widehat{\pi''}(t)} e_j 
  = F_{\pi'}(t) - \sum_{j\in \widehat{\pi''}(t)} e_j 
  \\
  &
  = D_{\pi''}(t).
 \end{align*}
 Since we consider every possible partial schedule~$\pi'$ for~$\Jset_{W}$,
 correctness follows.
\end{proof}
}
\cref{prop:mccco} yields that
\qeAcr{} is in \FPT{} \wpb{} the size~\mccco{} of a smallest modulator to~$c$-$\cco$ combined with the flexibility~$\phi$
and in~\XP{} \wpb{} \mccco{} alone.
A well-known modulator yielding constant-size connected components is a vertex cover, 
i.e., 
a modulator to $1$-$\cco$.
Since we can compute a minimum-size vertex cover in FPT-time (folklore),
we have that
 \qeAcr{} is \FPT{} \wpb{} the flexibility combined with vertex cover size of the job graph.
%

%
%
%
%

%
%
%
%
%

%
%
\subsection{W-hardness and Kernelization Lower Bounds Results}
\label{sec:pchardness}

Achieving \FPT{} for~$\mx{\flex}$ even in quite restricting settings or improving \cref{obs:concom} to polynomial kernelization seems unlikely.

\begin{theorem}[\appref{prop:NPhard:BP:rsame}]
 \label{prop:NPhard:BP:rsame}
 \qeAcr{} without battery is strongly \NP-hard and \W{1}-hard \wpb{}~$\tau$, 
 and hence also when parameterized by~$\flex$,
 even if $\mx{\ell}=1$, $\nrl = \ndl = 1$, and~$\im(F)=1$.
\end{theorem}

\appendixproof{prop:NPhard:BP:rsame}
{
Intuitively,
we cut off the first~$n$ time steps from the constructed instance in the proof of~\cref{prop:NPhard:BP:rdiff}
and set all release dates to the then resulting first time step.

\begin{proof}
 Let $I=(X = \{x_1, \dots, x_n\},\nbin,\sbin)$ be an instance of~\ubpTsc{}.
 We construct an instance~$I'$ of \qeAcr{} without battery as follows.
 Let $\tau \ceq \nbin$
 and $F(t)  = \sbin$ for every $t \in \set{\tau}$.
 For each $i\in \set{n}$, 
 add a job $J_i$ with $r_i = 1$, 
 $d_i = \tau$, 
 $\ell_i =1$, 
 and $e_i = x_i$.
 We claim that~$I$ is a \yes-instance if and only if~$I'$ is a \yes-instance.
 We have that~$x_i \in X_t \iff \pi(J_i)=t$.
 Thus,
 we have that all jobs are scheduled if and only if~$(X_1,\dots,X_{\nbin})$ is a partition.
 Moreover,
 we have~$\sum_{x_i\in X_t} x_i \leq \sbin \iff F(t) - \sum_{i\in \jis(t)} e_i \geq 0$.
\end{proof}
}

\begin{theorem}[\appref{prop:NOPK:cco}]
 \label{prop:NOPK:cco}
 \UnlessPK,
 \qeAcr{} without battery admits no problem kernel of size polynomial in~$\flex+\cco$,
 even if $\mx{\ell}=1$.
\end{theorem}
\appendixproof{prop:NOPK:cco}
{
\begin{proof}
 For an \NP-hard problem~$L$,
 a polynomial equivalence relation~$\calR$
 is a relation such that
 whether any two instances belong to the same equivalence class can be decided in time polynomial in their aggregated size,
 and any finite set of instances is partitioned by~$\calR$ into classes
 whose number is polynomially upper-bounded by the maximum instances size found in the set. 
 An \ANDcroco{} from an \NP-hard problem~$L$ into a parameterized problem~$L'$
 takes~$q$ $\calR$-equivalent instances~$I_1,\dots,I_q$ from~$L$
 and constructs an instance~$I=(x,k)$ of~$L'$ 
 in time polynomial in~$\sum_{p=1}^q |I_{p}|$ 
 such that~$k$ is polynomially upper-bounded by~$\max_{1\leq p\leq q} |I_{p}|+\log(q)$,
 and $I$ is a \yes-instance if and only if each of~$I_1,\dots,I_q$ is a \yes-instance.
 A parameterized problem that has an \ANDcroco{} admits no problem kernel of polynomial size \unlessPK{}~\cite{Drucker15}.
 
 Let~$I_1,\dots,I_q$ be $q$ instances of \ubpTsc{} each with~$k$ bins,
 i.e., $I_p=(X_p,\nbin,\sbin_p)$ with~$X_p=\{x_1^p,\dots,x_{n_p}^p\}$. 
 We construct an instance~$I$ without battery of \qeAcr{} as follows. 
 For each each~$x_i^p\in X_p$,
 construct a length-1 job~$J_{i}^p$ with release time~$(p-1)\cdot k+1$, deadline~$p\cdot k$,
 and energy~$x_i^p$.
 The forecast~$F$ is defined as follows:
 For every~$p\in\set{q}$,
 $F((p-1)\cdot k+t)=\sbin_p$ for every~$1\leq t\leq k$.
 We have that~$\flex+\cco{}\in O(\max_{1\leq p\leq q} |I_p|)$.
 We claim that~$I$ is a \yes-instance if and only if each if~$I_p$, $p\in\set{q}$ is a \yes-instance.
 The correctness follows the same lines as in the proof of~\cref{prop:NPhard:BP:rsame}.
\end{proof}
}

\section{ILP and Experiments}
\label{sec:experiments}
\appendixsection{sec:experiments}

In this section,
we present our experimental evaluation.
Since our theoretical analysis revealed that \qeAcr{}
is intractable already in quite restricted settings,
we employ integer linear programming (ILP) as solver.
Exploiting the power of ILPs,
we even solve a generalization of \qeAcr{},
where an external power source~$\ex\colon T\to \Qnn$ exists.
We include the external power source by replacing~\eqref{eq:Dpit}
with
$D_{\pi}(t) = F(t) + \ex(t)- E_{\pi}(t)$ for all~$t\in T$.

\dectask{OPT-\qeTsc{} (OPT-\qeAcr)}{opttdv}
{a forecast~$F$, a set~$\calJ$ of jobs, and a battery~$\Bat$}
{find a schedule~$\pi$ and external energy~$\ex$ making~$\pi$ feasible such that~$\sum_{t\in T} \ex(t)$ is minimized}

The minimum sum of external energy can be understood as distance-to-autarky measure.
Note that \qeAcr{} is OPT-\qeAcr{}
where no external power source is available ($\ex(t) = 0$ for all~$t$).
Thus,
\qeAcr{} reduces to OPT-\qeAcr{} and hence the latter generalizes the former.

\subsection{ILP Model}

Consider the following ILP,
called \qeAcr{}-ILP,
that minimizes the sum of external power:
\begin{align}
 \min \sum\nolimits_{t\in T} X_t, \quad X_{t} \in \Qnn \,\forall t\in T. \label{ilp:goal}
\end{align}
First,
we model the job scheduling.
For each~$j\in N$,
let~$T_j\ceq \{r_j,\dots,d_j\}$ and we have:
\begin{align}
 x_{j,t}&\in \{0,1\} \ \forall\,t\in T_j,\quad  z_{j,t}\in \{0,1\} \ \forall\,t\in S_j \label{ilp:constr:defvarxz}
 \\
 1 &= \sum\nolimits_{t\in S_j} z_{j,t}  
 , \qquad
 x_{j,t} = \sum\nolimits_{s\in S_j: s\leq t\leq s+\ell_j-1} z_{j,s} \ \ \forall\,t\in T_j \label{ilp:constr:starttime}
\end{align}
Note that~$z_{j,t}$ models whether job~$J_j$ starts at time step~$t\in S_j$
and~$x_{j,t}$ models whether job~$J_j$ is active at time step~$t\in T_j$.
Herein, \eqref{ilp:constr:starttime} ensures that each job starts and the active times
of a job form a consecutive sequence starting in accordance with~$z_{j,t}$.

Next, we model the battery including the external power.
For all~$t\in T_{+1}\ceq \set[2]{\tau+1}$,
we have the following:
\begin{align} 
 0&\leq B(t) \leq b_c, \qquad B(1) = \bz, \qquad B(\tau+1) \geq \bz \label{ilp:constr:battery:basics}
 \\ 
 E_{t-1}^+ &\leq F(t-1) + \ex_{t-1} + E_{t-1}^- 
 - \sum\nolimits_{j\in N:\:(t-1)\in T_j} x_{j,t-1}\cdot e_j  \label{ilp:constr:net}
 \\
 0 &\leq E_{t-1}^+ \leq y_{t-1} \cdot \bl / \lin, \qquad\quad y_{t-1}\in\{0,1\} \\
 0 &\leq E_{t-1}^- \leq (1 - y_{t-1}) \cdot \bc \cdot \lout,
 \label{ilp:constr:net:adjust}
 \\ 
 B(t) &= B(t-1) + \lin\cdot E_{t-1}^+ - \lout^{-1}\cdot E_{t-1}^- \label{ilp:constr:batteryupdate}
\end{align}

Note that 
we additionally require in~\eqref{ilp:constr:battery:basics} that at each end of the scheduling horizon,
the battery's level is at least~$\bz$.
This makes the initial battery level~$\bz$ meaningful when computing on consecutive days
(cf.\ \cite[Sec.~2.3]{MunzkeSBH21}).
Moreover,
our ILP formulation deviates from a formulation of \eqref{eq:net} and \eqref{eq:bat:update}
that directly implements $\min$ and $\max$ functions
and enforces the fully available excess to go into the battery.
We model the $\max$-$\min$ constellation of \eqref{eq:net} through lower bounds of zero and the binary variables~$y_{t-1}$;
The nested $\min$ in \cref{eq:bat:update} we model by upper bounding the battery sizes by~$\bc$,
the excess by $\bl/\lin$,
and the demand by $\bc\cdot \lout$. 
Finally,
note that in \eqref{ilp:constr:net},
the excess~$E_{t-1}^+$ is only upper bounded.
This allows that the battery is possibly loaded with less than the actual available excess
and
only as much as required to reach an overall minimum of external energy.
Indeed, 
we will obtain several ILP solutions making use of this
(see \cref{fig:example_result} in the appendix).
While these formulation tricks are quite simple,
we observe a significant effect: 
we tested a more direct formulation 
which failed to compute solutions with larger flexibilities in reasonable time.

Finally,
note that we have a 1-to-1 correspondence between
a solution for OPT-\qeAcr{},
where we additionally require~\eqref{ilp:constr:battery:basics},
and an optimal solution for \qeAcr{}-ILP
via~$\pi(J_j)=t \iff z_{j,t}=1$.

\begin{remark}
  \label{rem:poweroftheilp}
  Our ILP allows easily for many reasonable modifications.
  If one is interested in minimizing the cost of the energy,
  one can set the goal function to
  $\min \sum\nolimits_{t\in T} w_t\cdot X_t$,
  where~$w_t$ is the energy's price in time step~$t$.
  When two jobs~$i$ and~$j$ must be disjoint,
  we add $x_{i,t}\neq x_{j,t}$ for all~$t\in T_i\cap T_j$.
  When a job~$i$ has to start or end before another job~$j$,
  we can add $\sum_{t\in T_j} t\cdot z_{i,t} < \sum_{t\in T_{j}} t\cdot z_{j,t}$  or 
  $\ell_i + \sum_{t\in T_j} t\cdot z_{i,t} \leq \sum_{t\in T_{j}} t\cdot z_{j,t}$,
  respectively.
  When each job may have several,
  mutually disjoint horizons in which it is allowed to be scheduled,
  then consider their union for \eqref{ilp:constr:starttime}.
  Our data includes no information about any such further requirements addressed above.
  Hence, these constraints are omitted.
  \rqed
\end{remark}

\subsection{Experiments}

In this section,
we describe the data, setup, and results of our experiments.
The goal of our experiments is threefold.
First, we quantify the practical impact of flexibility on reducing external energy under realistic solar radiation forecasts.
Second, we compare small-scale balcony PV systems with residential rooftop systems to understand how system scale influences achievable autarky.
Third, we investigate which structural properties of instances---such as job density and forecast irregularity---drive the computational difficulty of solving OPT-\qeAcr{} in practice.

\subsubsection{Data}

\subparagraph*{Solar Radiation Data.}

To simulate our forecast,
we use data 
from the Institute for Electrical Information Technology
of the Technical University Clausthal,
Germany,
for radiation (in Watt per square meter [\vwattpersqm{}]) measured for each minute [\vmin{}] for every year from 2016 to 2022.
The measurement is at a fixed point 
in Clausthal-Zellerfeld, 
Germany.

\subparagraph*{Appliance-level Power Consumption Data.}

We use data 
\cite{Alhamoud2014goSMART,Alhamoud2015EWSN}
in which for two households~A and B
in
Germany,
several devices,
such as kettle, coffee maker, fridge, or oven,
were monitored over a certain time period at a temporal resolution of a few seconds.
We only use the data for household~A since household~B's data is too sparse.
We set a threshold of \vwatt{6} for a device to be active,
and two active points in time that are at most 120 seconds apart 
belong to the same job.
Each job then has a start and end time
(understood as release date and deadline at flexibility 0),
given at minute resolution.
The difference between them defines the job's length.
The job's energy requirement is defined as the maximum observed power (in W) over its execution interval
(recall that this ensures the job to be executable even if energy demand fluctuates).
There is no information about the flexibility of any job in the data.

\subsubsection{Setup}

\subparagraph*{PVs.}
For PVs,
we assume a threshold of \vwattpersqm{10},
i.e., 
below which we set an entry to zero.
We assume a pessimistic-realistic efficiency of 20\%~\cite{Green56,Shaik2023}.
We consider two scenarios,
called~\pvfnt{RA} and \pvfnt{BS}.
\pvfnt{RA} (\pvfnt{BS}) corresponds to a rooftop (balcony) PV,
where
we assume \vsqm{48} (\vsqm{3}) of PV area~\cite{loedl2010pvpotenzial,ringel2024germanypv} directed south.
This allows us to compare plug-in balcony systems with residential rooftop systems~\cite{KRASCHEWSKI2025115092}.
So,
given the radiation data and the sun's position per time step,
the above described transformations are performed to compute the actual forecast.

\subparagraph*{Jobs.}
We computed the minimum external energy for an \emph{operational day},
where an operational day
starts at 4 am and ends at 3:59 am on the next day
(at minute resolution, leading to 1440 time steps).
\begin{figure}
 \centering
 \includegraphics[width=0.995\textwidth]{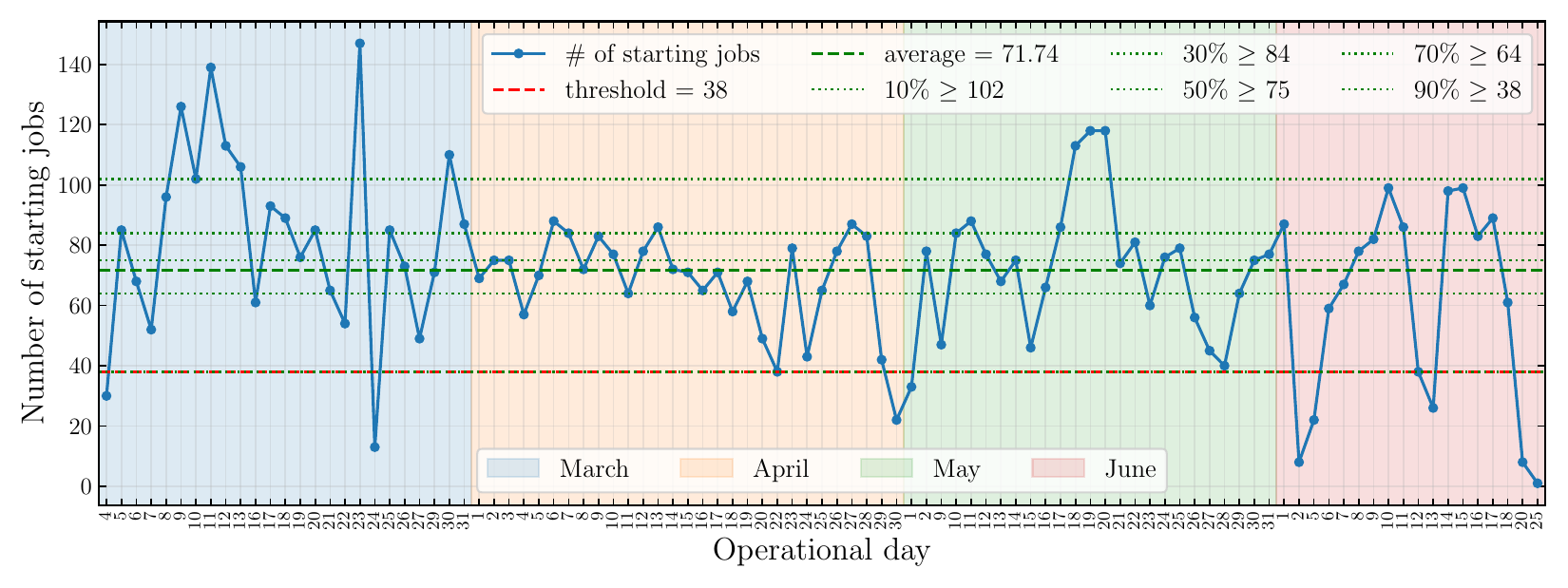}
 \caption{Overview of the number of jobs starting per operational day (calendar date) from March to June. 
 The x-axis shows the day of the month 
 (e.g., 2 in the green area is May~2). 
 Percentiles of daily job counts are included; 
 for instance, 
 on 70\% of the operational days, at least 64 jobs start.}
 \label{fig:jobs}
\end{figure}
We selected the 90\% (88) of all available days between March and June
(see \cref{fig:jobs} for an overview)
with the highest number of jobs
to cut off small instances,
leading to each instance having at least 38 jobs
(the average of all days is about 72 jobs).
For this set of jobs,
the average daily mean job lengths is $\vmin{39.17}$
and the average daily median job lengths is $\vmin{6.88}$.
We added flexibility scenarios next to the 0\% scenario given by the data.
These are specified through a parameter~$\hat{\flex}=0.25\cdot x$,
$x\in\set[0,1]{12}$,
corresponding to a symmetric extension of up to 300\% of the original time window length.
Concretely,
for $\hat{\flex}=1$ and $\hat{\flex}=3$,
a job whose length equals the average daily mean job length obtains additional time windows of $\vmin{19.58}$ and $\vmin{58.75}$,
respectively,
before and after its original release date.
For a job whose length equals the average daily median job length,
the corresponding additional time windows are $\vmin{3.44}$ and $\vmin{10.32}$,
respectively.
Formally,
if~$r_j$, $d_j$, and~$\ell_j$ are the release date, deadline, and length of~$J_j$ obtained from the data (recall that~$\ell_j=d_j-r_j+1$ here),
then~$r_j' = \max\{1,r_j-\hat{\flex}\cdot \ell_j/2\}$ and $d_j' = \min\{\tau,d_j + \hat{\flex}\cdot \ell_j/2\}$
(here, 
$\tau=1440$, 
corresponding to the last time step for the operational day).
We took the forecast of every year from 2016 to 2022,
and combined them with the jobs recorded 
on the respective days.

\subparagraph*{Battery.}
The battery is measured in Watt-minutes [\vwattmin{}].
We assume a battery with capacity \vkilowattmin{60}
(\vkilowatthour{1}),
commonly considered for balcony-sized PVs~\cite{BARZEGKARNTOVOM20201302},
with a loading speed of 3 hours from 0\% to 100\%
and an initial battery load of 10\%.
Further,
we assume
that the effective in- and output efficiencies are equal
($\lambda\ceq \lin=\lout$)
and consider three values~$\lambda\in\{0.9,0.94,0.98\}$
around the commonly assumed value 0.95~\cite{STENZEL2018165}.
Since we focus on the effect of forecast variability and flexibility, 
we fix the battery size to isolate 
these factors.

\subsubsection{Results}

We used Python 3.10 and Gurobi 12.0 (Python interface)
to compute solutions for~$\qeAcr$-ILP.\footnote{%
Run on Intel\textsuperscript{\textregistered} Xeon\textsuperscript{\textregistered} Silver 4310 CPU@2.10GHz (12 cores),
125GB RAM,
Ubuntu 22.04.3 LTS (x86\_64).
}
We compared \pvfnt{RA} and \pvfnt{BS},
each available day from 2016--2022,
the thirteen flexibility values,
and the three effectivities of our battery.
In total,
we solved 24024
instances for each of \pvfnt{RA} and \pvfnt{BS}.
\cref{fig:example_result} (appendix) shows%
\toappendix{%
  \begin{figure*}[t!]
  \centering
  \includegraphics[width=1\textwidth]{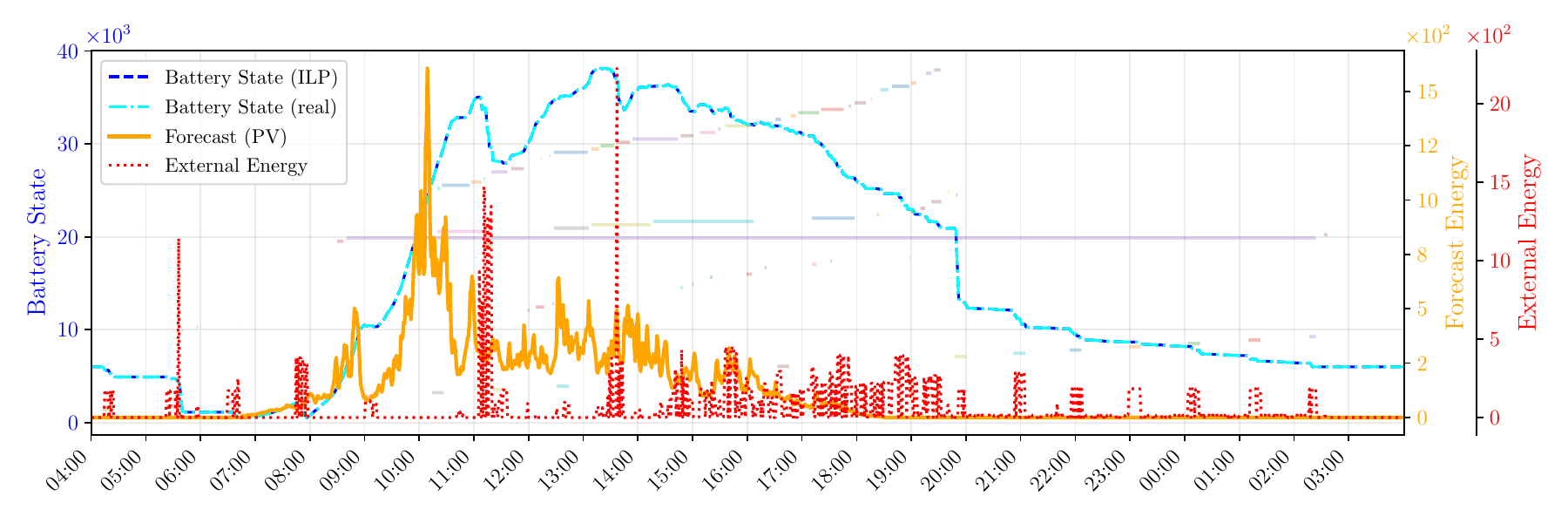}
  \includegraphics[width=1\textwidth]{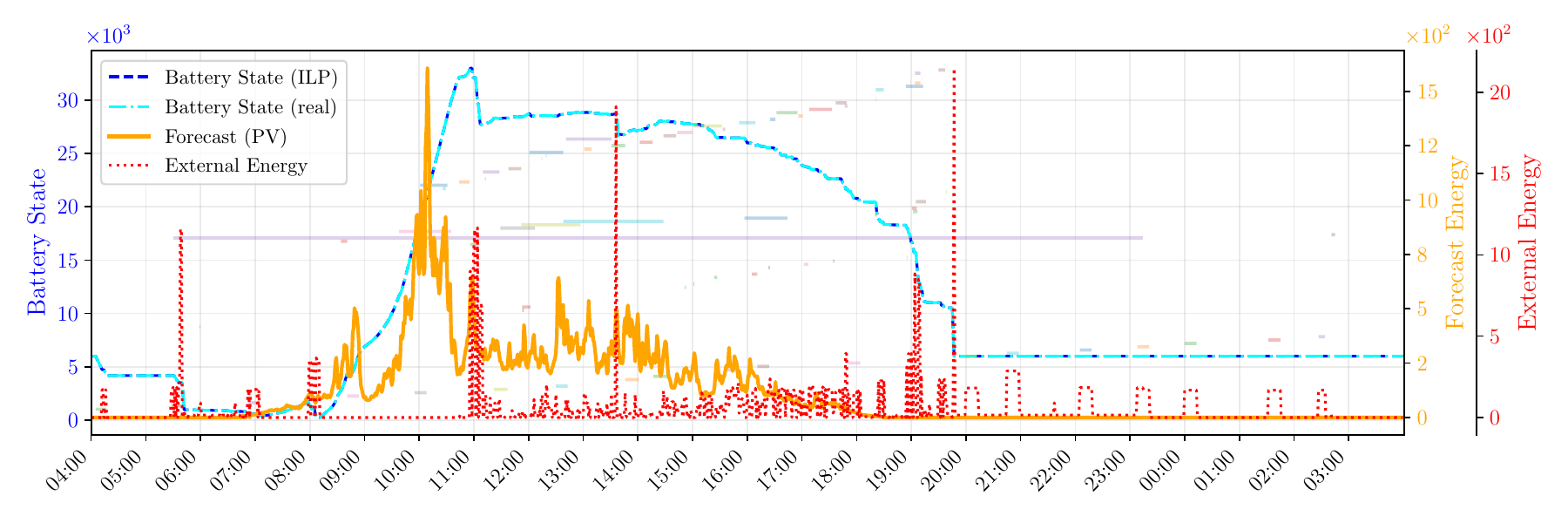}
  \includegraphics[width=1\textwidth]{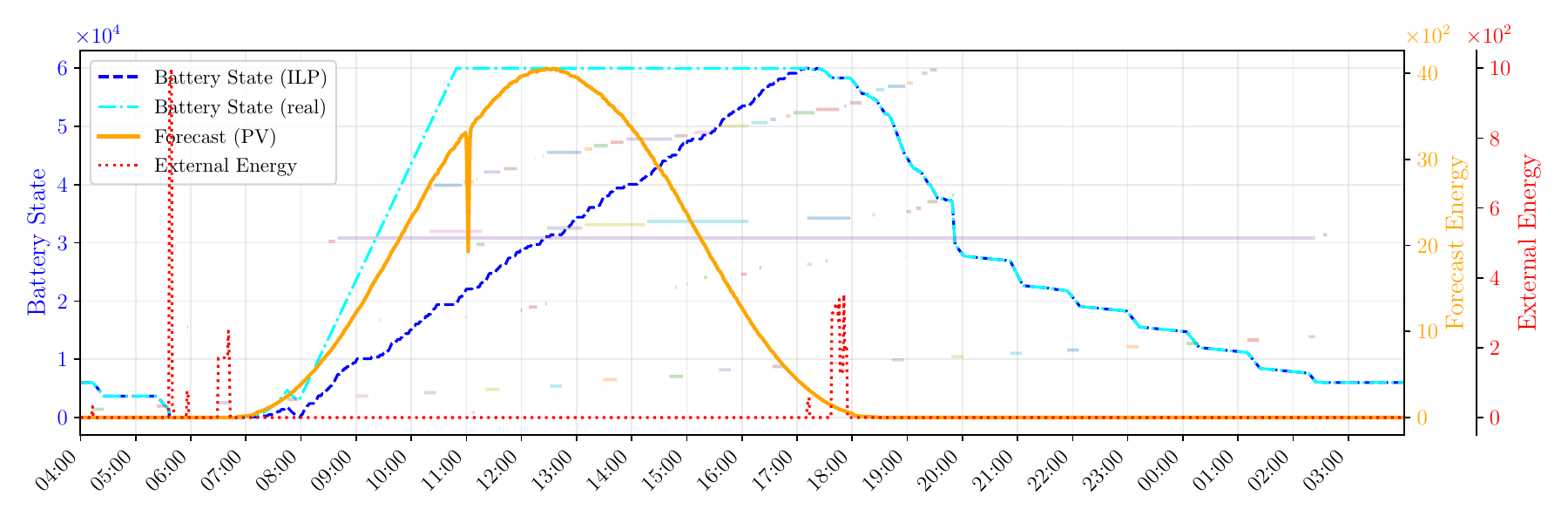}
  \includegraphics[width=1\textwidth]{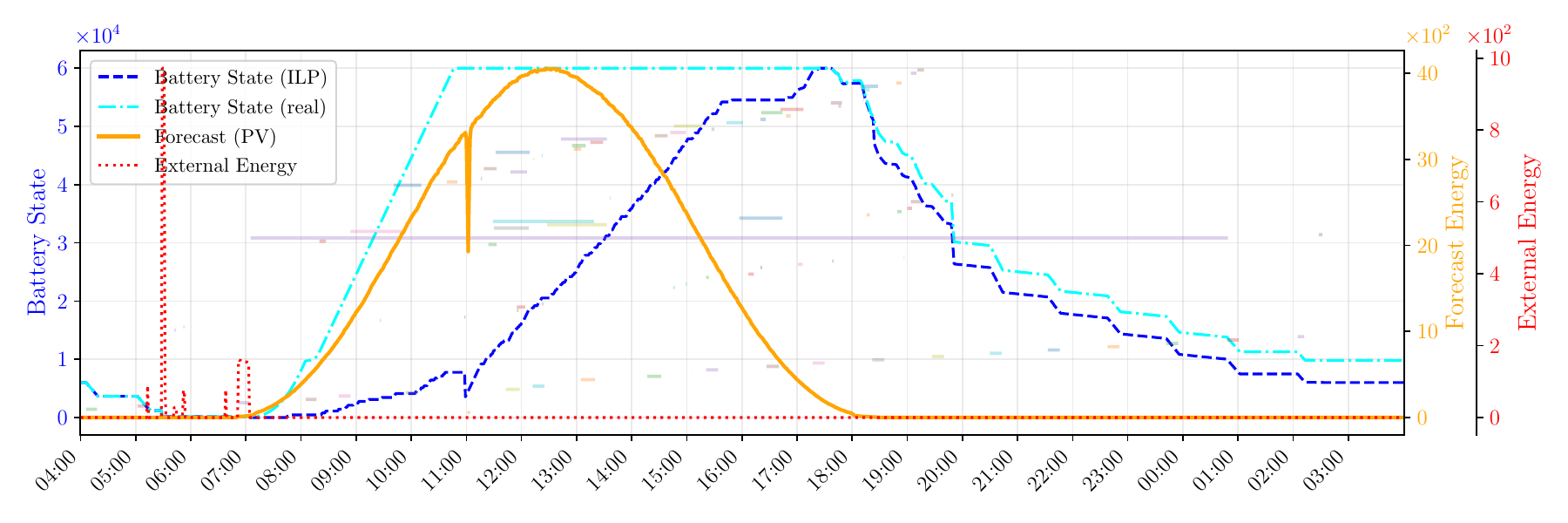}
  \caption{Illustration of ILP solution 
  for~\pvfnt{RA}
  with $\lambda = 0.94$
  for operational day March 13 in year 2020 (upper half)
  and in year 2022 (lower half).
  In each half,
  the upper and lower plot corresponds to
  flexibility 0 and 3,
  respectively.
  The total required external energies
  for flexibility 0 and 3
  are 84210 respectively 75497 for 2020,
  and 9876 respectively 5275 for 2022.
  Jobs are depicted as colored stripes,
  each at its unique $y$-value (corresponding to the job's index). 
  We present both battery states:
  those computed by the ILP,
  and the ``real'' ones in the sense of when the jobs are scheduled according to the ILP solution.
  }
  \label{fig:example_result}
  \end{figure*}
}%
optimal solutions for four example instances.
When applicable,
we discuss the results for~$\lambda=0.94$ in greater detail.

\subparagraph*{Energy Reduction.}
In \cref{fig:relgain},
\begin{figure}[t]
 \centering
 \includegraphics[width=0.495\textwidth]{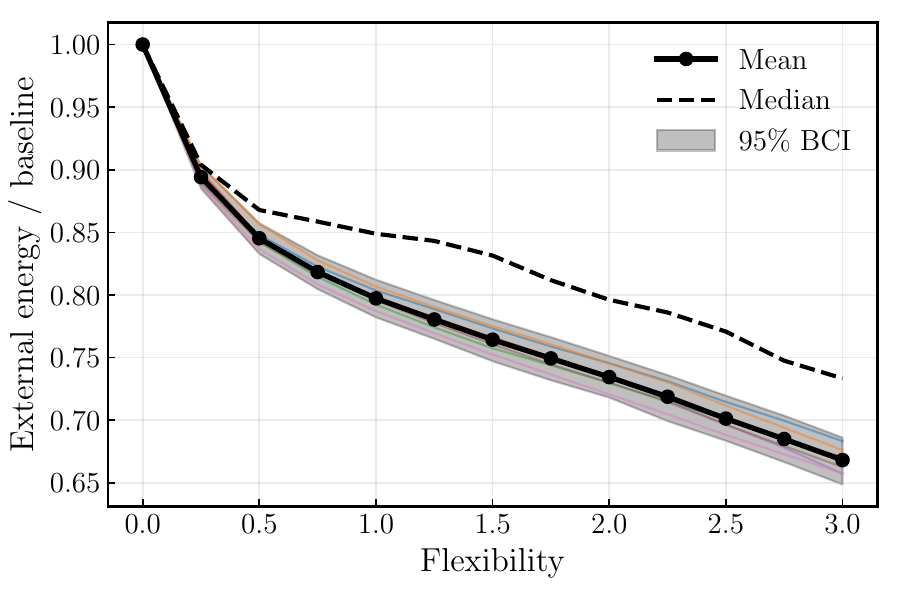}
 \hfill
 \includegraphics[width=0.495\textwidth]{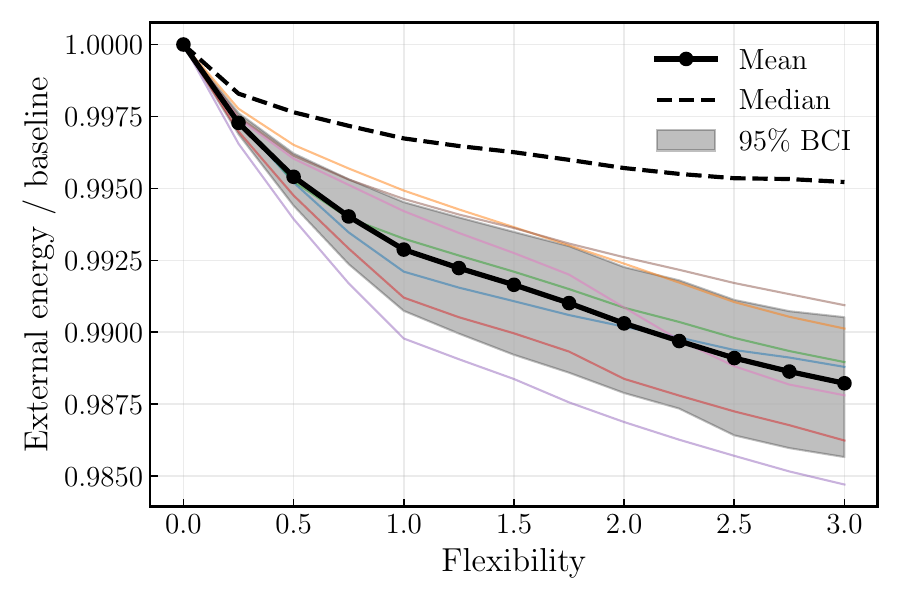}
 \caption{Flexibility versus aggregated mean and median external energy reduction 
 over all years 2016--2022 and~$\lambda=0.94$
 with 95\% BCI for (left) \pvfnt{RA} and (right) \pvfnt{BS}. 
 Each year's mean is shown in a different color 
 (used for illustration only).}
 \label{fig:relgain}
\end{figure}
we show the external energy reduction against the bottom line of flexibility zero
when increasing the flexibility,
aggregated over all years 2016--2022 
for~$\lambda=0.94$.
Our results
show that for both \pvfnt{RA} and \pvfnt{BS},
a higher flexibility allows higher energy savings on average.
With our maximum flexibility of 3,
these are relatively small for \pvfnt{BS} with a reduction of $1.18\%$ on average 
(median: $0.48\%$),
yet quite large for \pvfnt{RA} with about $33.18\%$ on average 
(median: $26.66\%$).
Notably,
the improvements are larger in the flexibility interval~$(0,1]$ 
and are then almost linear in $(1,3]$.
While the standard deviations are quite high (cf.\ \cref{fig:relextenergy}),
the 95\% percentile nonparametric bootstrap confidence interval~\cite{EfronT94} with 2000 replications (BCI for short) 
is quite narrow,
in particular for~\pvfnt{RA} ($\hat{\flex}=3$: $[0.65, 0.69]$).
Regarding the battery effectivity,
we observe two orthogonal trends.
With increasing effectivity,
the reduction for~$\hat{\flex}=3$
slightly decreases for \pvfnt{BS}
($\lambda=0.9$: $1.55\%$; $\lambda=0.98$: $0.83\%$)
but slightly increases for \pvfnt{RA}
($\lambda=0.9$: $32.8\%$; $\lambda=0.98$: $33.51\%$).
Overall, 
flexibility yields substantial benefits for rooftop systems (with small-space batteries), 
but only marginal improvements for small balcony systems.

We additionally investigated which additional factors the energy reduction
along flexibility may co-depend
(see~\cref{fig:relextenergy,fig:relextenergy-aggregated} for year 2022).%
\toappendix
{
\begin{figure*}[t]
 \centering
 \includegraphics[width=0.328\textwidth]{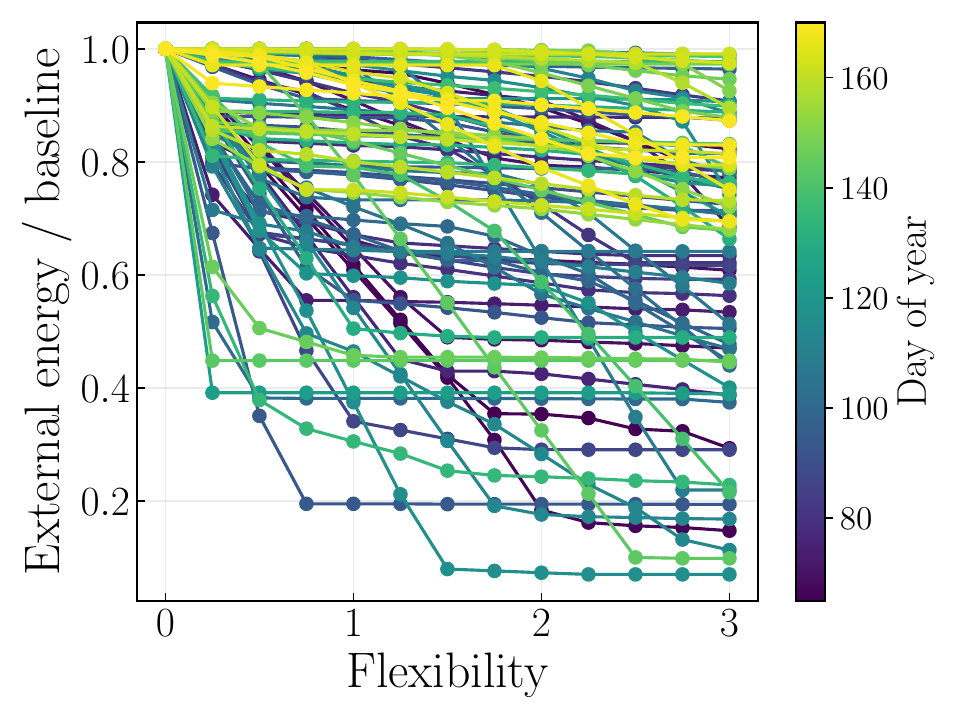}
 \includegraphics[width=0.328\textwidth]{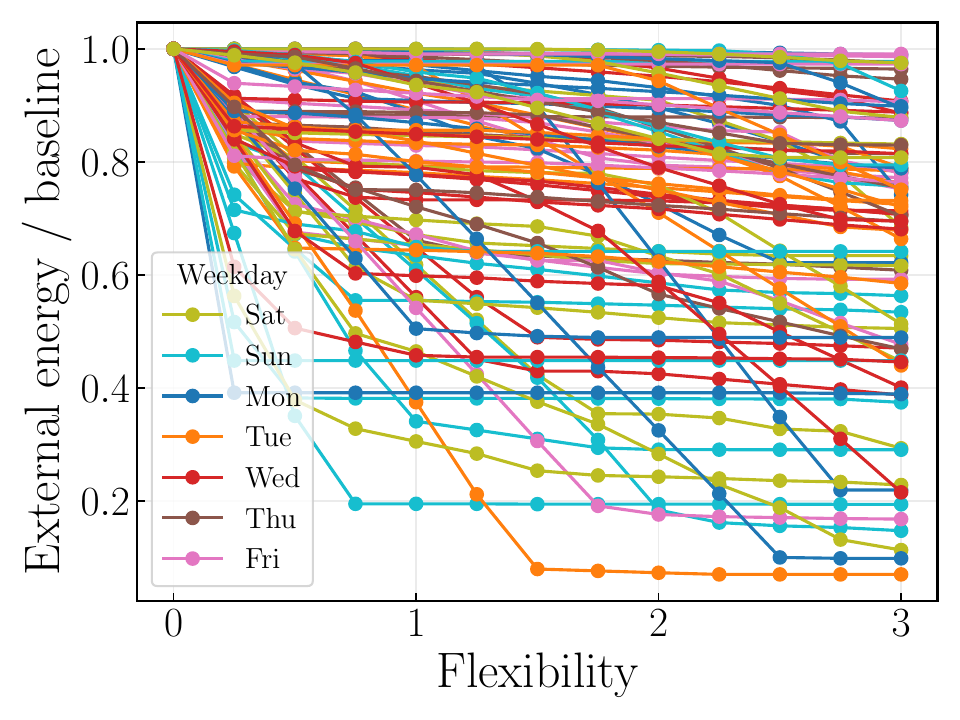}
 \includegraphics[width=0.328\textwidth]{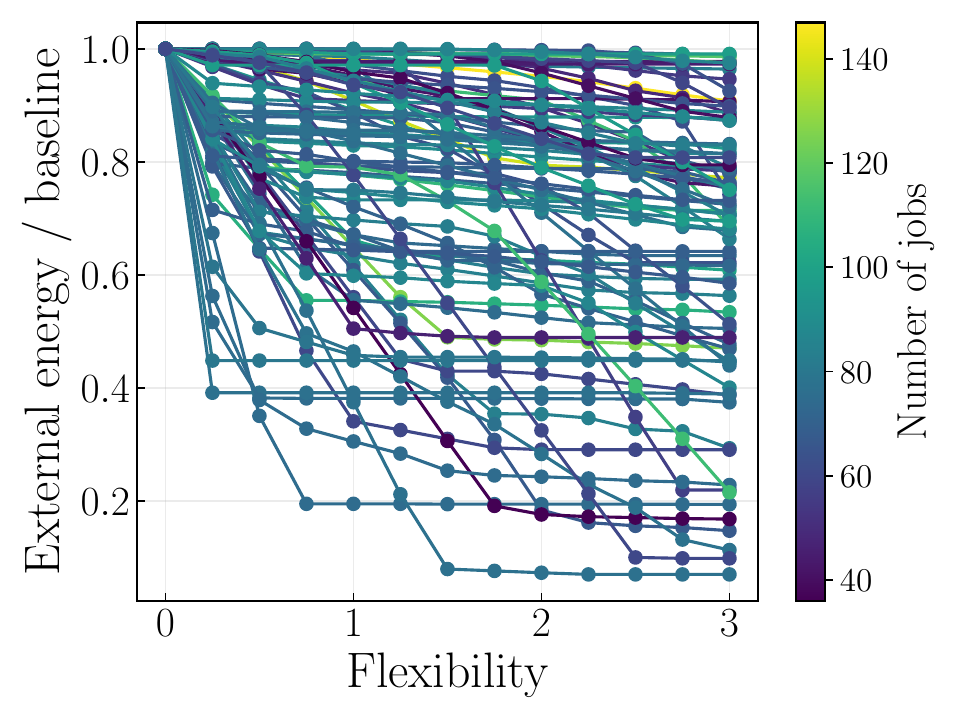}
 \includegraphics[width=0.328\textwidth]{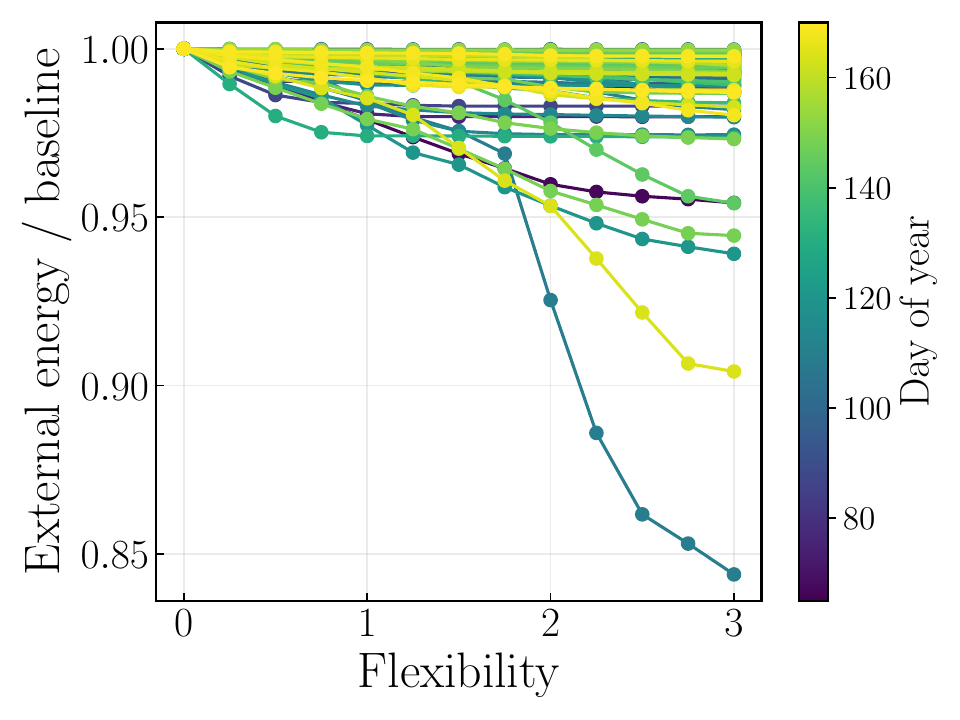}
 \includegraphics[width=0.328\textwidth]{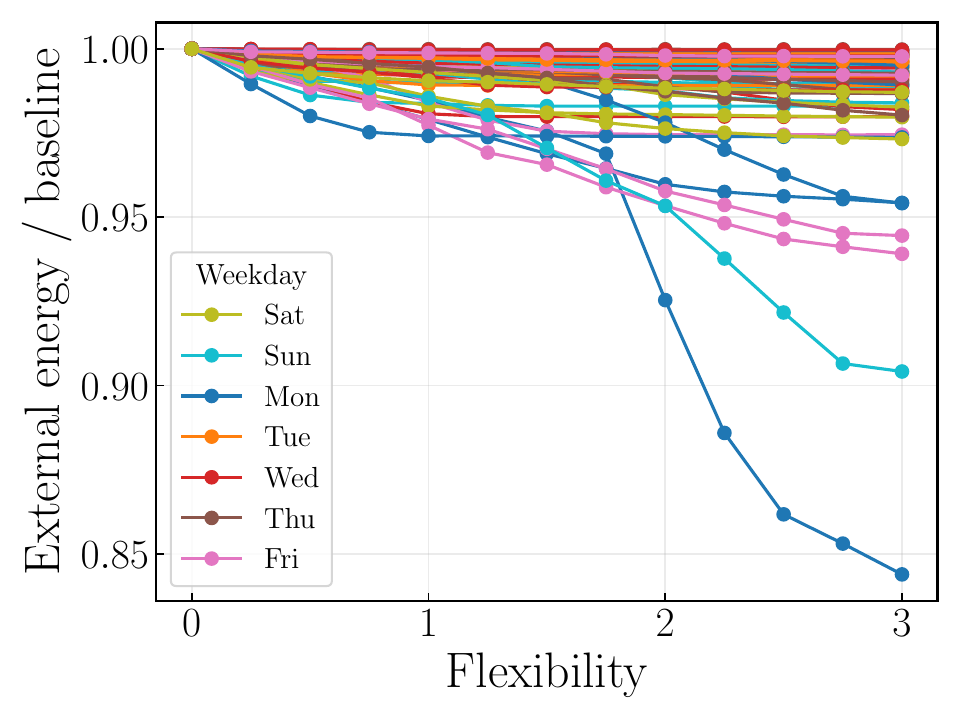}
 \includegraphics[width=0.328\textwidth]{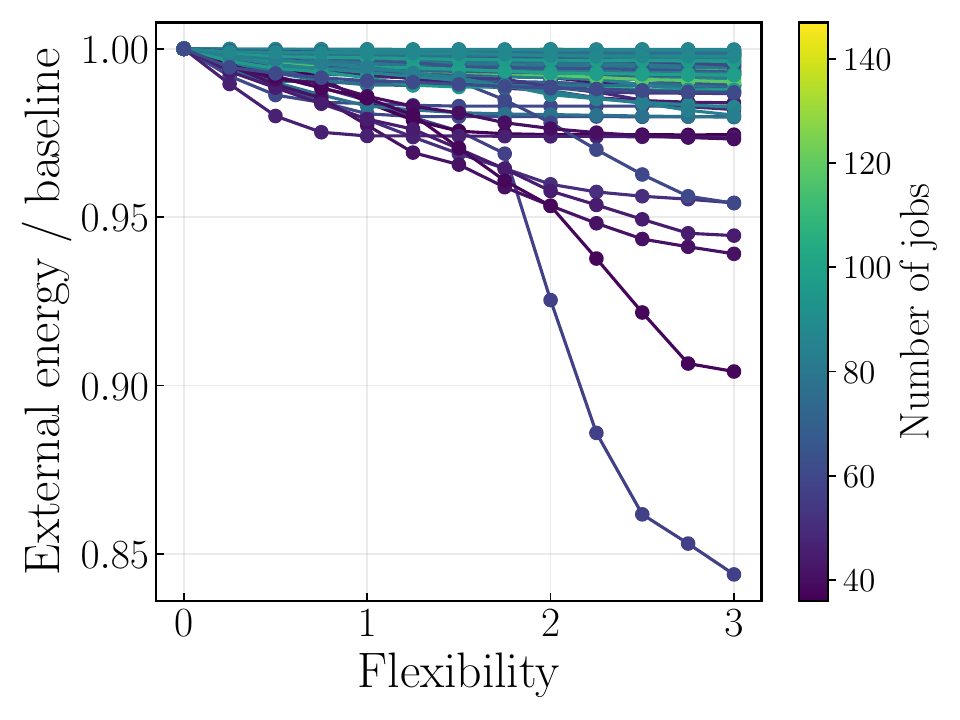}
 \caption{External energy reduction per flexibility for all instances in the year~2022 with~$\lambda=0.94$ for (top) \pvfnt{RA} and (bottom) \pvfnt{BS}.
 Instances are colored by day of the year (left),
 weekday (middle),
 and number of jobs (right).
 }
 \label{fig:relextenergy}
\end{figure*}
}
\begin{figure*}[t]
 \centering
 \includegraphics[width=0.328\textwidth]{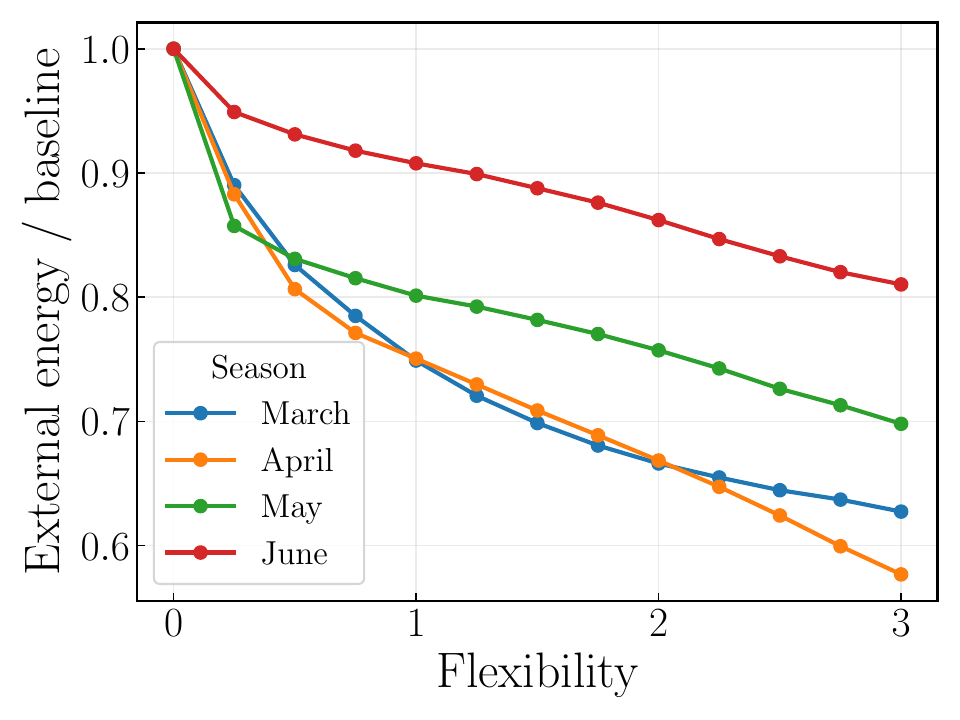}
 \includegraphics[width=0.328\textwidth]{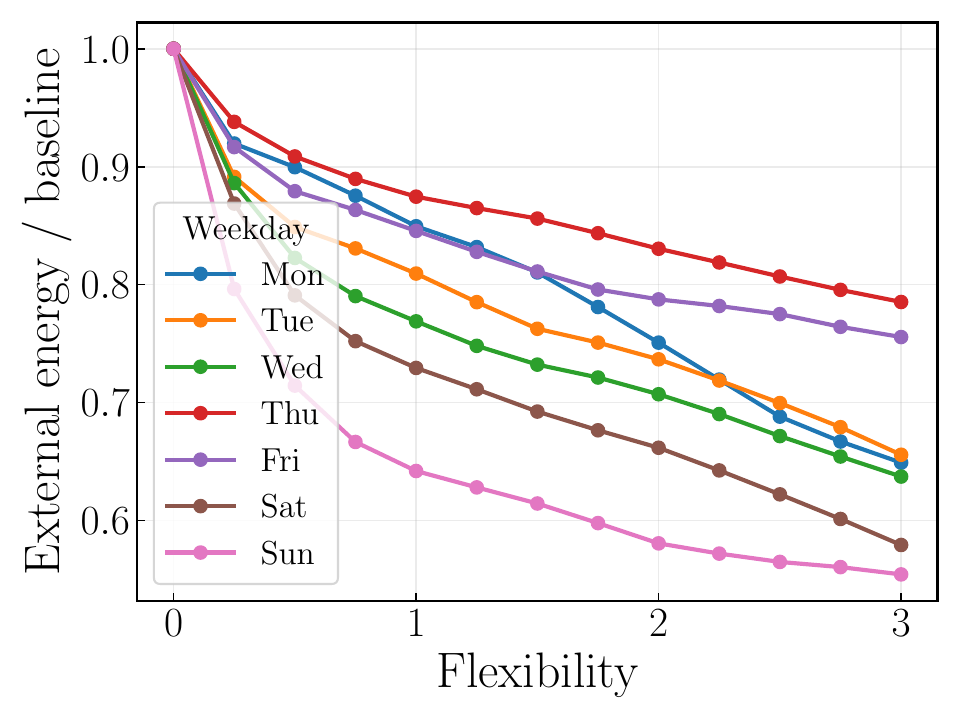}
 \includegraphics[width=0.328\textwidth]{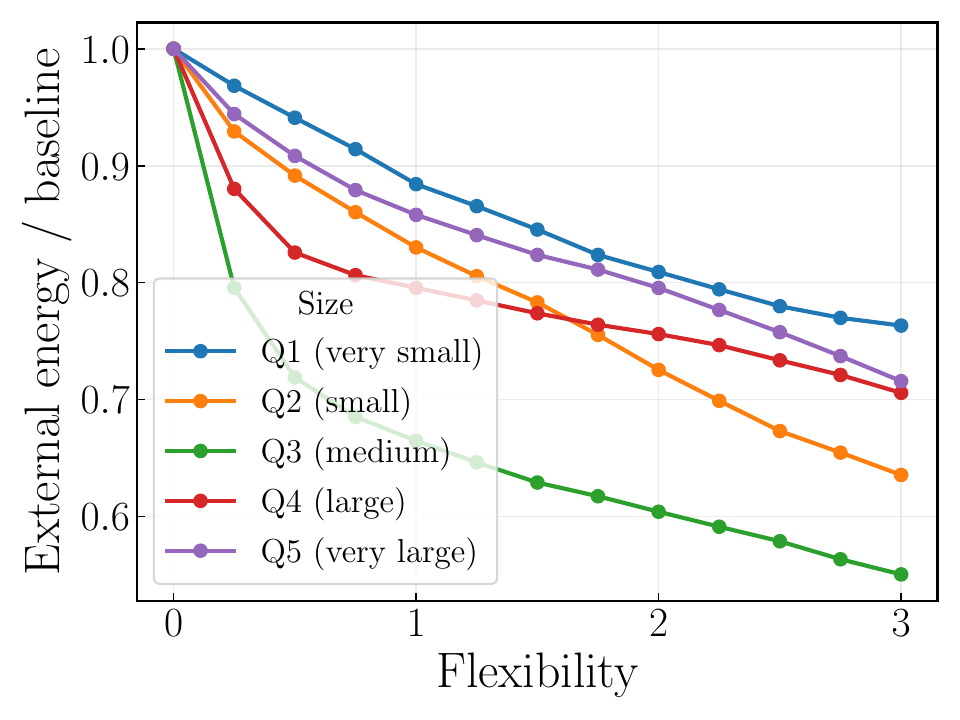}
 \includegraphics[width=0.328\textwidth]{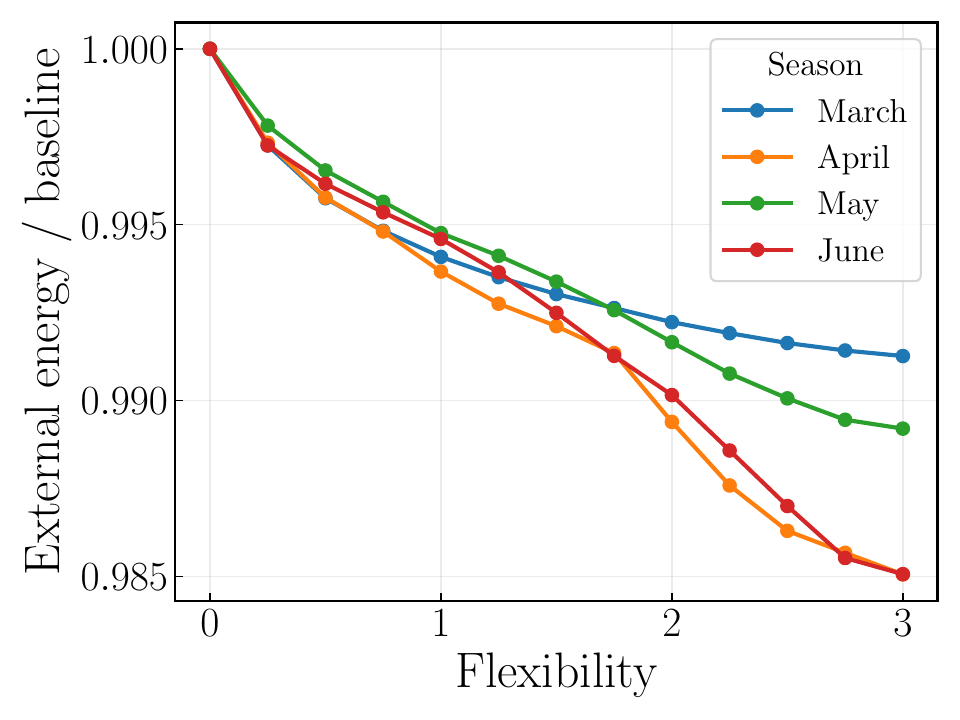}
 \includegraphics[width=0.328\textwidth]{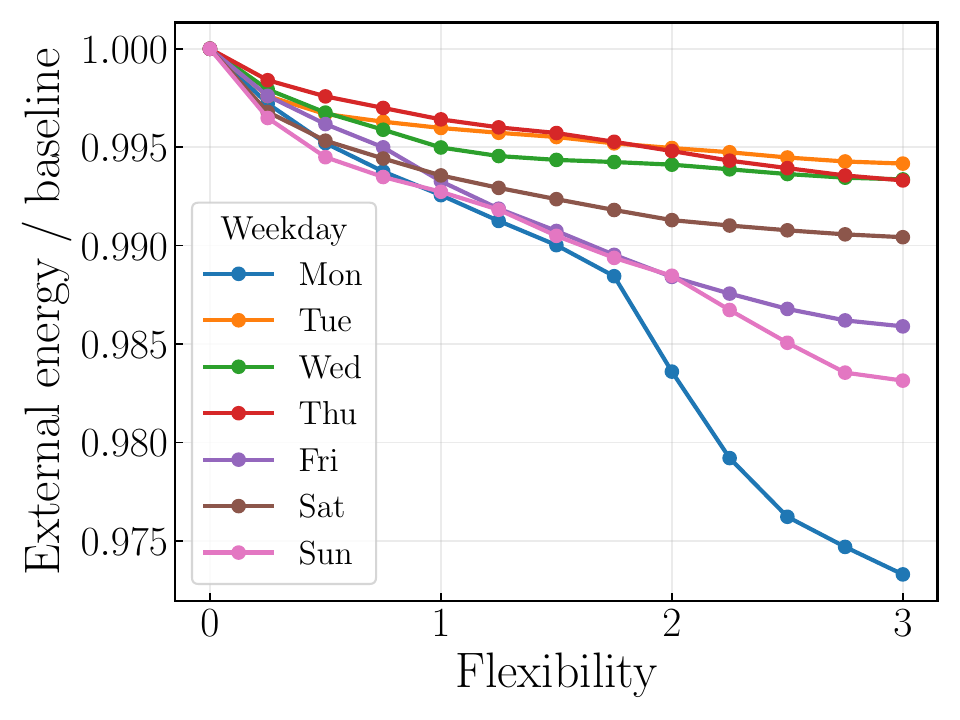}
 \includegraphics[width=0.328\textwidth]{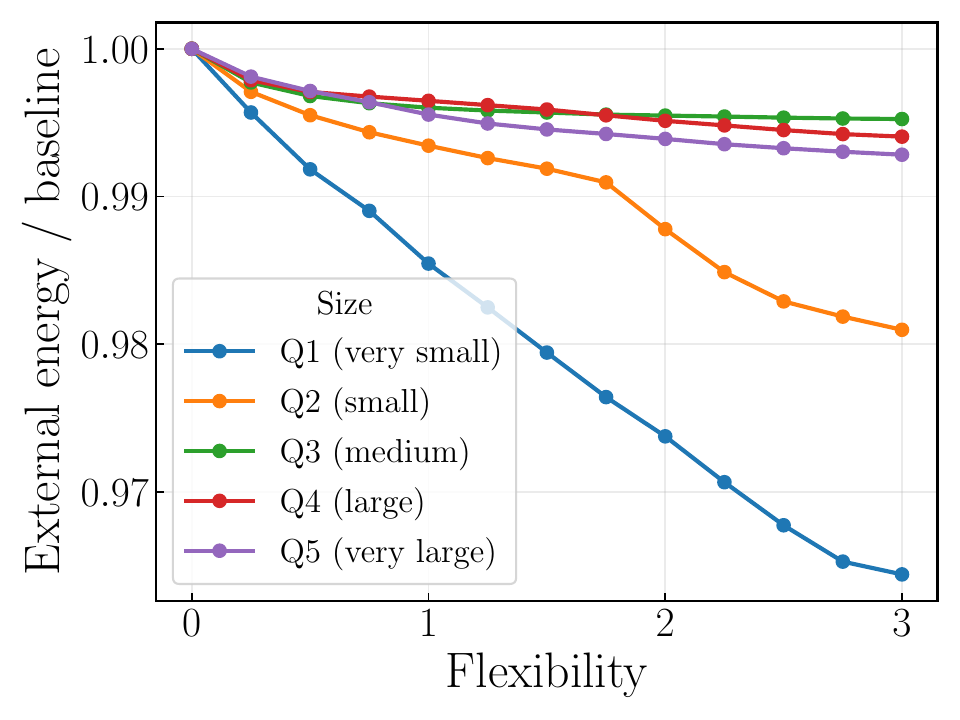}
 \caption{External energy reduction per flexibility for all instances aggregated 
 in the year~2022 with~$\lambda=0.94$ for (top) \pvfnt{RA} and (bottom) \pvfnt{BS}.
 Instances are aggregated by month (left),
 weekday (middle),
 and five percentiles of the number of jobs (right).
 }
 \label{fig:relextenergy-aggregated}
\end{figure*}
We checked for a seasonal correlation (by day of year),
a weekday correlation,
or correlation with the number of jobs.
For neither of these factors,
we could identify a clear correlation.

\subparagraph*{Running Times.}

The runtime increases with the flexibility almost linearly
(see~\cref{fig:runtimes}),
\begin{figure}[t]
 \centering
 \includegraphics[width=0.495\textwidth]{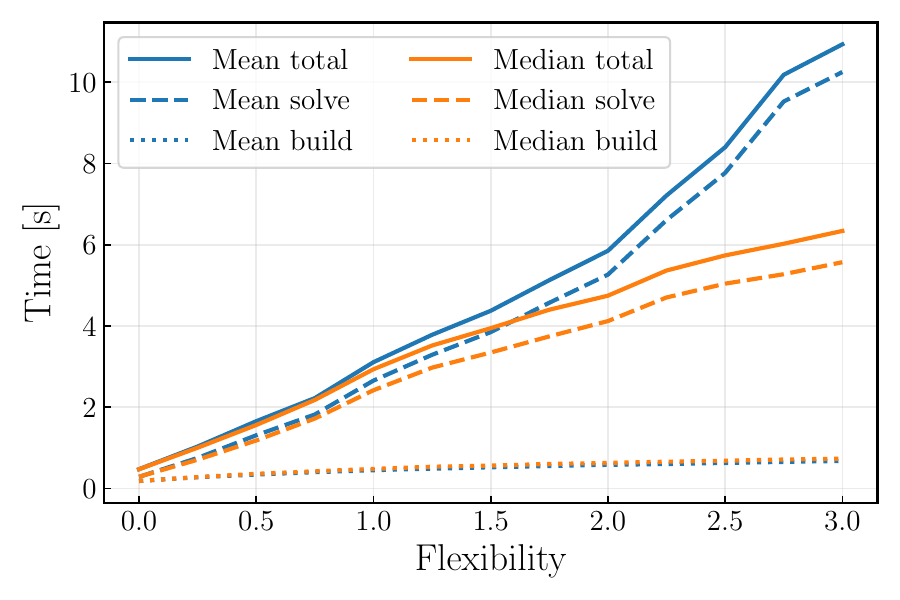}
 \hfill
 \includegraphics[width=0.495\textwidth]{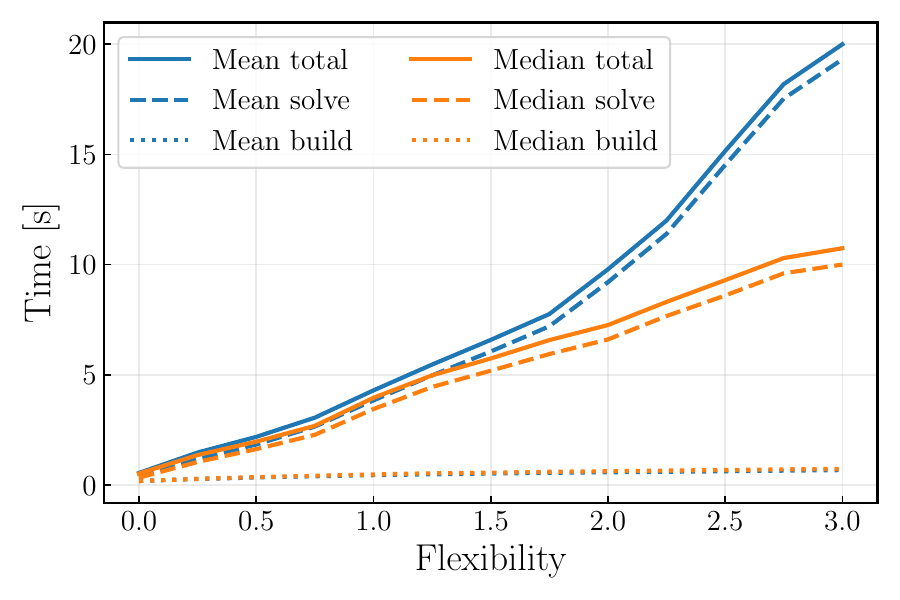}
 \caption{Flexibility versus mean and median runtimes for \pvfnt{RA} (left) and \pvfnt{BS} (right). 
 Here,
 'build' and 'solve' refer to the runtimes of constructing and solving the ILP,
 respectively.}
 \label{fig:runtimes}
\end{figure}
up to
\vsec{10.928} (\vsec{6.338})  and \vsec{19.989} (\vsec{10.742}) on average (median)
at~$\hat{\flex}=3$ for \pvfnt{RA} and \pvfnt{BS},
respectively.
Note that there are few yet severe outliers
when the flexibility is at least 2
(see~\cref{fig:runtimes-byflex}):
\begin{figure}[t]
 \centering
 \includegraphics[width=0.495\textwidth]{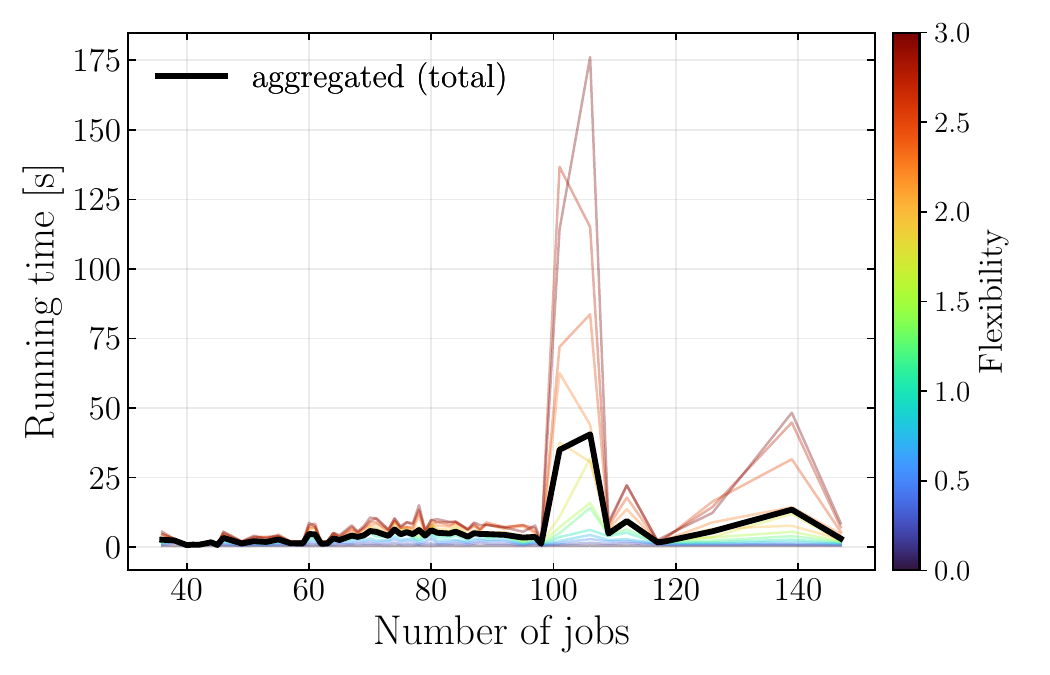}
 \hfill
 \includegraphics[width=0.495\textwidth]{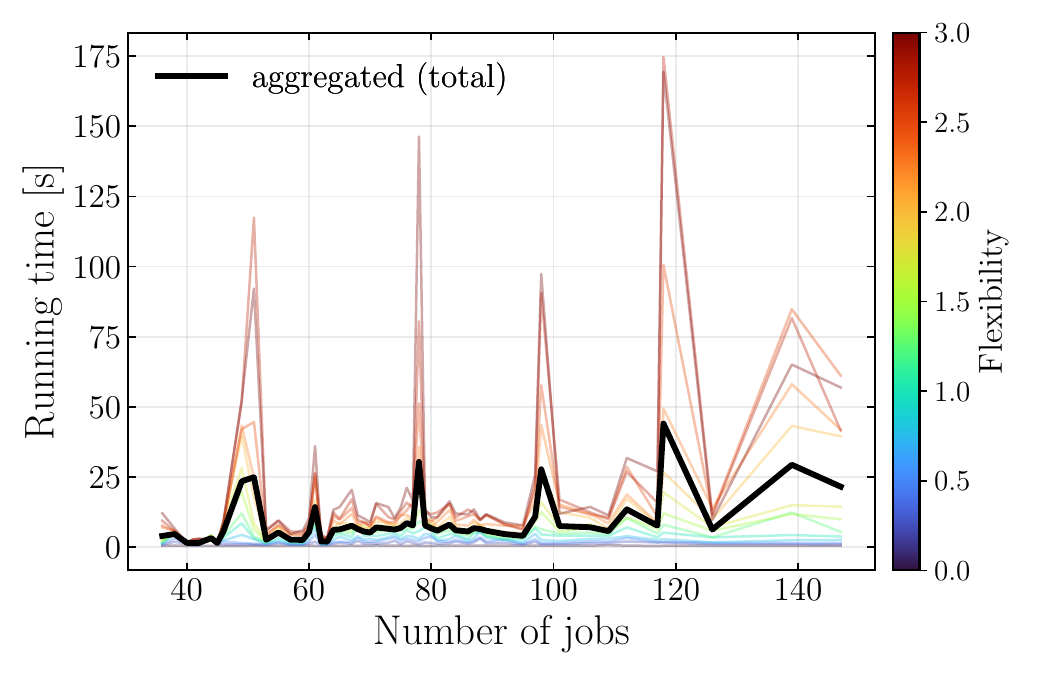}
 \caption{Running times versus number of jobs versus flexibility,
 with \pvfnt{RA} (left) and \pvfnt{BS} (right).}
 \label{fig:runtimes-byflex}
\end{figure}
Here, in the range of 100--120 jobs,
it may take up to 3 minutes on average to compute solutions for~$\hat{\flex}=3$.
We also detected that the runtime increases with increasing number of scheduled jobs
(see \cref{fig:runtimes-corrs});%
{
\begin{figure}[t]
 \centering
 \includegraphics[width=0.495\textwidth]{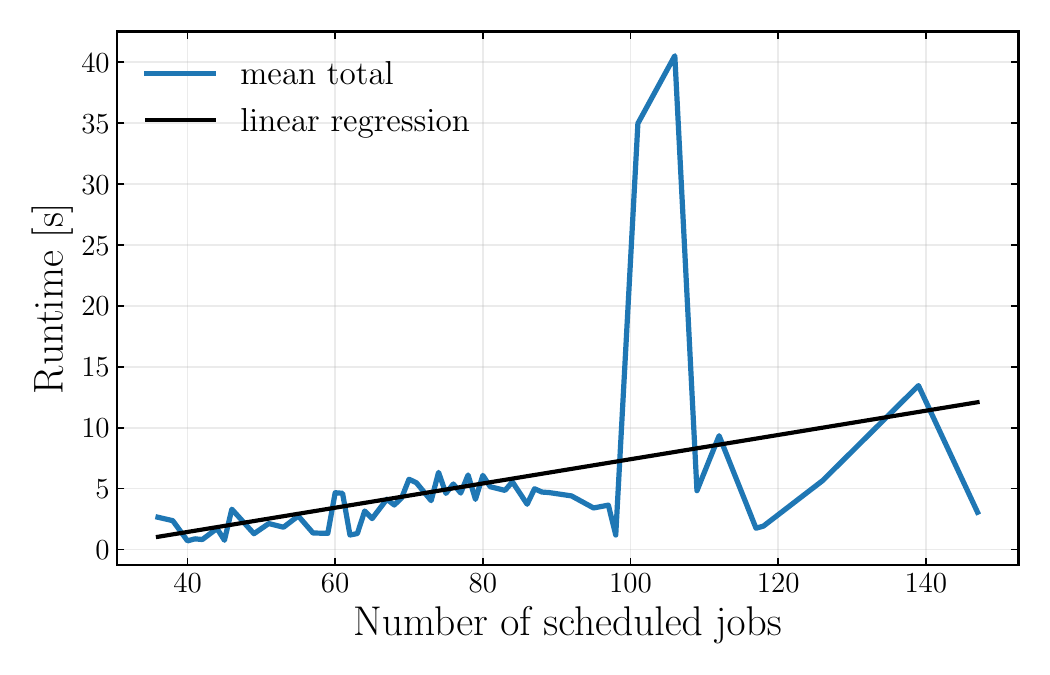}
 \hfill
 \includegraphics[width=0.495\textwidth]{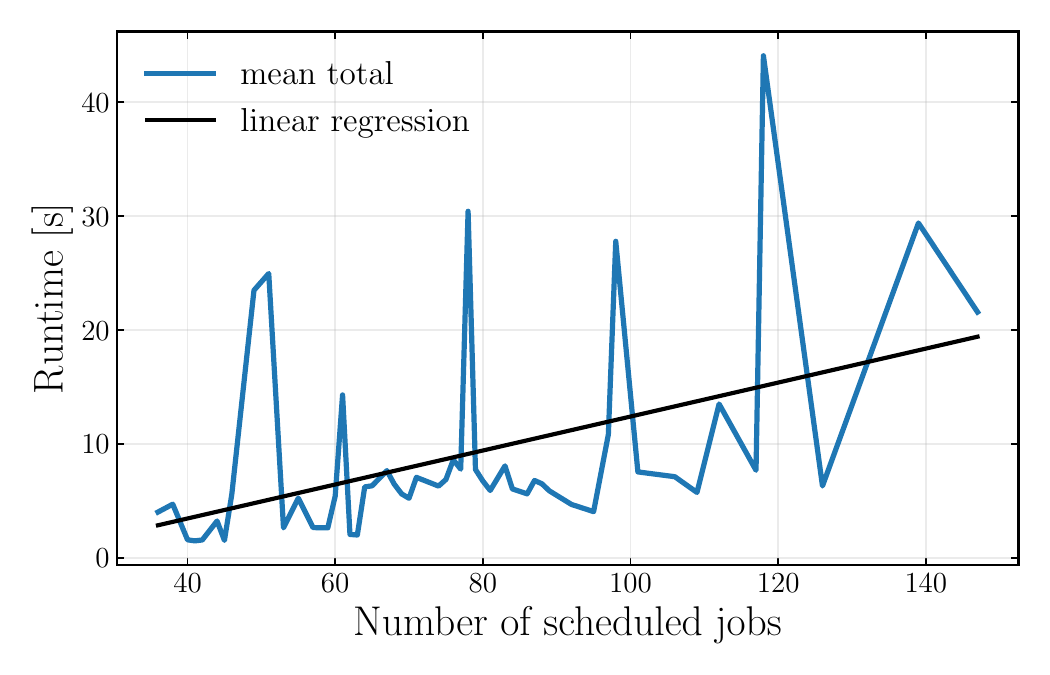}
 \includegraphics[width=0.495\textwidth]{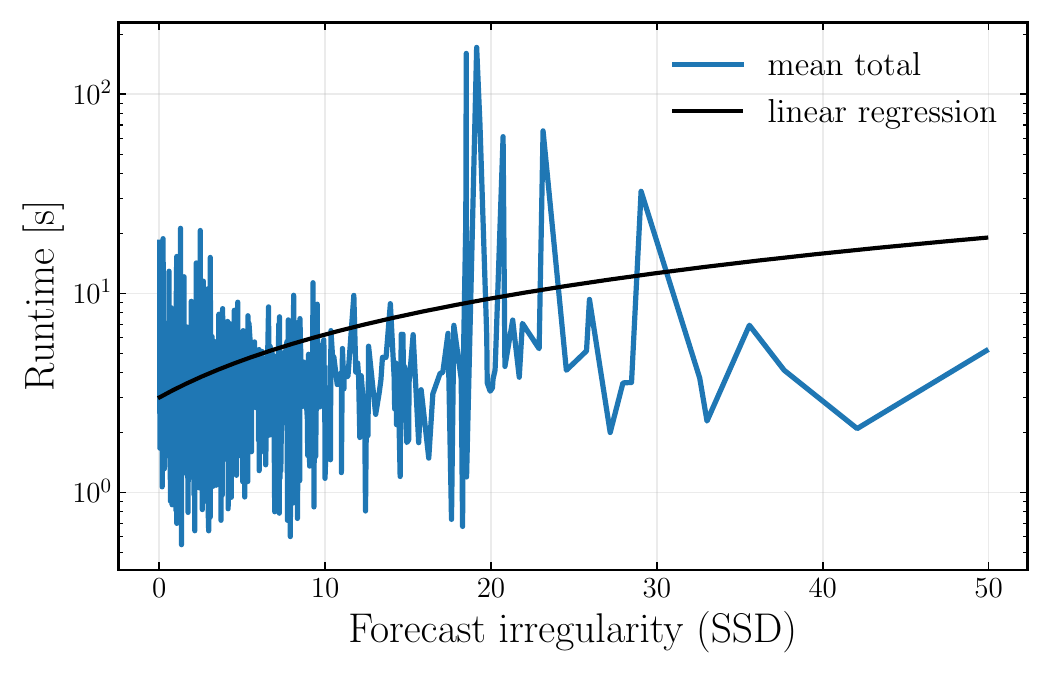}
 \hfill
 \includegraphics[width=0.495\textwidth]{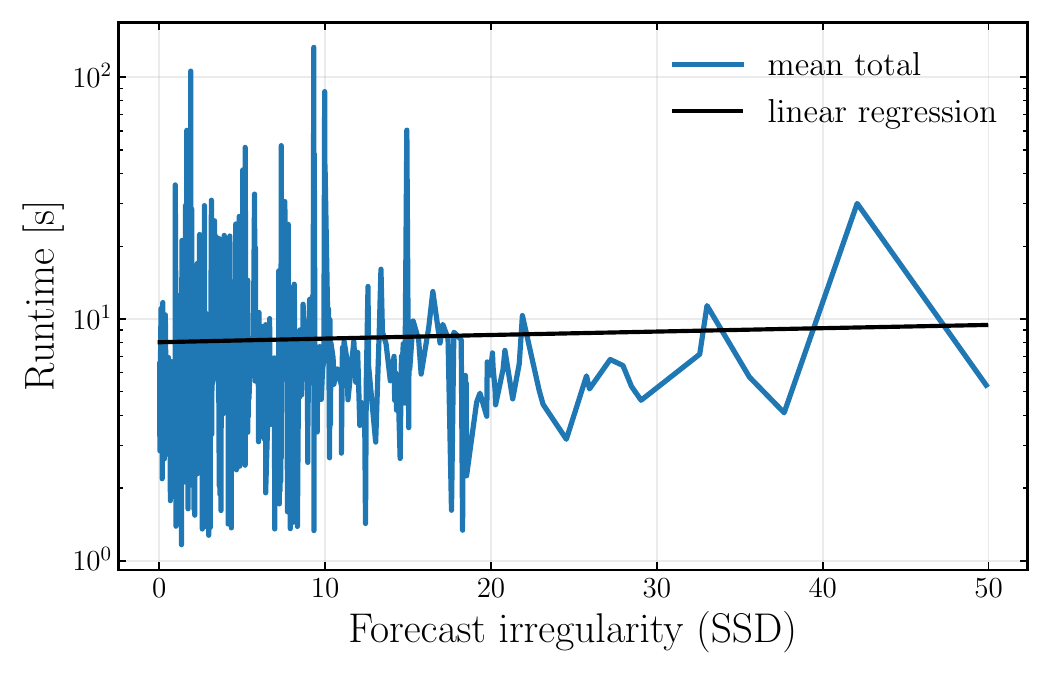}
 \caption{Running times versus (top) number of jobs and (bottom) forecast irregularity ($y$-$\log$ scale),
 with \pvfnt{RA} (left) and \pvfnt{BS} (right).}
 \label{fig:runtimes-corrs}
\end{figure}
}%
this effect was even more significant for \pvfnt{BS}
(with an~$r$-value of 0.4480, 
as compared to 0.3723 for \pvfnt{RA}).
In contrast,
we observed a weak positive correlation
($r$-value of 0.1871) 
for \pvfnt{RA} 
showing that the runtime increases with increasing \emph{forecast irregularity}.
We measure the forecast irregularity of an operational day by 
first normalizing the forecast by a clear-sky forecast,
and then taking the sum of the squared differences.
Interestingly, our hardness results rely on structurally simple forecasts.
In contrast, the empirical runtimes increase with forecast irregularity, 
indicating that temporal variability is a practically relevant complexity driver.
We point out that for \pvfnt{BS},
we did not detect such a correlation with statistical significance.
Finally,
we identified no correlation between the runtime and the battery effectivity
for either \pvfnt{RA} or \pvfnt{BS}.

\section{Conclusion}
\label{sec:epilogue}

We introduced and analyzed a novel scheduling problem motivated by energy autarky, 
combining parameterized complexity theory with data-driven evaluation. 
Our results reveal a sharp contrast between worst-case computational hardness and substantial empirical gains achievable through flexibility.
Overall, 
our experiments confirm that flexibility substantially reduces external energy in larger PV systems, 
already for relatively small flexibility values,
while also revealing how forecast irregularity and job density may influence computational effort.
For more detailed discussions on the theoretical and experimental parts,
as well as a concluding big picture,
see below.

\subparagraph*{Computational and Parameterized Complexity.}

Our most intriguing open question is whether 
the weak versus strong \NP-hardness boundary is substantial.
In this context,
combining~$\mx{\flex}$ with the number of distinct job energies remains open.
Further open questions regard whether \qeAcr{}
is \FPT{} \wpb{} the number~$n$ of jobs alone.
We know that guessing only the job order is insufficient (\cref{prop:horiz:joborders}),
and
that a battery is required for hardness (\cref{prop:nobattery:FPTn}).
For our \FPT{} results,
e.g., for the parameter~$\mx{\flex}+n$,
we wonder whether polynomial-sized problem kernels exist.
Finally, future work may uncover additional polynomial-time solvable cases by adapting the algorithms from~\cref{prop:greedy,prop:oneortwouniquerels}.

\subparagraph*{Experimental Evaluation.}
Assuming symmetric flexibility on both ends of a job's time window may be unrealistic in practice,
but we are unaware of data connecting jobs with their typical flexibilities.
Thus,
we seek additional and structurally diverse datasets to validate and extend our analysis.
Such data could also enhance our understanding of which additional factors significantly impact the energy reduction and runtimes.
In our experiments,
we considered \vhour{24} time windows starting at 4~am.
For most households,
4~am is a reasonable cutoff for defining operational days with respect to executed tasks.
Future work may investigate the robustness of our results under shifted,
shorter, 
or longer time windows.
As to the battery,
in our experiments
we assumed that the in- and output effectivity is equal. 
While this is practically a reasonable assumption~\cite{STENZEL2018165},
it would be interesting to investigate how the results change with larger gaps between the two
effectivities.
Moreover,
studying how charge-dependent efficiencies affect the outcomes, 
compared with the constant-efficiency setting considered here,
is a well-motivated research direction.
For future work,
we plan to analyze the effect on the energy reductions of varying PV orientations~\cite{MubarakWS19} as well as battery loading speeds and sizes.
Regarding battery sizes,
we have preliminary results that indicate that doubling the battery size to \vkilowatthour{2} leads to a marginally higher mean runtime and with average reductions that are slightly higher for \pvfnt{RA} and lower for \pvfnt{BS}.
Moreover,
with a huge battery size of \vkilowatthour{6} the average reduction diminishes since mostly the required external energy is little to none.

\subparagraph*{Big Picture.}
Our work demonstrates that flexibility-aware energy scheduling forms a rich algorithmic problem at the intersection of sustainability and complexity theory.
We point out that the model, 
and hence our ILP, 
applies not only to households and PVs,
but also to small neighborhoods,
wind,
or other power sources;
it is not even restricted to electrical energy.
While discharge-rate limits were assumed to be non-critical at the household level,
they may have to be included for larger-scale or industrial instances
(cf.~\cref{rem:chargingspeed}).
We anticipate further applications of our model and corresponding algorithmic results.

\newpage
\bibliography{strings-long,quarter-energy-bib}

\clearpage
\appendix
\section*{Appendix}
\appendixProofText

\end{document}